\documentclass[sigconf]{acmart}
\usepackage{algorithm, pseudocode}
\usepackage{booktabs} 
\usepackage{graphicx}
\usepackage{amssymb}
\usepackage{amsfonts}
\usepackage{amsmath}
\usepackage{amsthm}
\usepackage{amssymb}
\usepackage{bm}
\usepackage{etex}
\usepackage{color}
\usepackage{graphicx}
\usepackage{esint}
\usepackage{tabularx}
\usepackage{multirow}
\usepackage[tight]{subfigure}
\usepackage{multirow}
\usepackage{mathtools}
\usepackage[noend]{algpseudocode}
\usepackage[colorinlistoftodos]{todonotes}
\usepackage{chngcntr}
\usepackage{titlesec}

\newtheorem{definition}{Definition}
\newtheorem{theorem}{Theorem}

\newtheorem{example}{Example}
\newtheorem{problem}{Problem}

\renewcommand{\P}{\mathbb{P}}
\newcommand{\K}{\mathbb{K}}
\newcommand{\F}{\mathbb{F}}
\newcommand{\Z}{\mathbb{Z}}
\newcommand{\N}{\mathbb{N}}
\renewcommand{\L}{\mathbb{L}}
\newcommand{\ang}[1]{\{#1\}}

\newcommand{\spa}{\textnormal{ }}

\newcommand{\frakf}{\mathfrak{f}}
\newcommand{\frakp}{\mathfrak{p}}

\DeclarePairedDelimiter\floor{\lfloor}{\rfloor}

\setcopyright{rightsretained}

\begin{document}
\copyrightyear{2019} 
\acmYear{2019} 
\setcopyright{acmlicensed}
\acmConference[ISSAC '19]{International Symposium on Symbolic and Algebraic Computation}{July 15--18, 2019}{Beijing, China}
\acmBooktitle{International Symposium on Symbolic and Algebraic Computation (ISSAC '19), July 15--18, 2019, Beijing, China}
\acmPrice{15.00}
\acmDOI{10.1145/3326229.3326256}
\acmISBN{978-1-4503-6084-5/19/07}
\title{Computing the Characteristic Polynomial of a Finite Rank Two Drinfeld Module}

\author{Yossef Musleh}
\affiliation{%
  \institution{Cheriton School of Computer Science \\ University of Waterloo}
  \city{Waterloo}
  \state{Ontario}
  \country{Canada}
}
\email{ymusleh@uwaterloo.ca}

\author{\'Eric Schost}
\affiliation{%
  \institution{Cheriton School of Computer Science \\ University of Waterloo}
    \city{Waterloo}
  \state{Ontario}
  \country{Canada}
}
\email{eschost@uwaterloo.ca}

\begin{abstract}
  Motivated by finding analogues of elliptic curve point counting
  techniques, we introduce one deterministic and two new Monte Carlo
  randomized algorithms to compute the characteristic polynomial of a
  finite rank-two Drinfeld module. We compare their asymptotic
  complexity to that of previous algorithms given by Gekeler,
  Narayanan and Garai-Papikian and discuss their practical
  behavior. In particular, we find that all three approaches represent
  either an improvement in complexity or an expansion of the parameter
  space over which the algorithm may be applied. Some experimental
  results are also presented.
\end{abstract}




\setcounter{secnumdepth}{4}
\renewcommand{\theparagraph}{\thesubsection.\arabic{paragraph}}
\counterwithin{paragraph}{subsection} % makes paragraph depend on subsection
\titleformat{\paragraph}[runin]{\normalfont\normalsize\bfseries}{\theparagraph.}{1em}{}
\titlespacing*{\paragraph}{0em}{1ex}{1em}
\newcommand{\pref}[1]{{\bf\ref{#1}}}

\maketitle

%%%%%%%%%%%%%%%%%%%%%%%%%%%%%%%%%%%%%%%%%%%%%%%%%%%%%%%%%%%%
%%%%%%%%%%%%%%%%%%%%%%%%%%%%%%%%%%%%%%%%%%%%%%%%%%%%%%%%%%%%
%%%%%%%%%%%%%%%%%%%%%%%%%%%%%%%%%%%%%%%%%%%%%%%%%%%%%%%%%%%%

\section{Introduction}

Drinfeld modules were introduced by Drinfeld in~\cite{Drinfeld74}
(under the name {\em elliptic modules}) to prove certain conjectures
pertaining to the Langlands program; they are themselves extensions of
a previous construction known as the {\em Carlitz
  module}~\cite{Carlitz35}.

In this paper, we consider so-called Drinfeld modules of rank two over
a finite field $\L$. Precise definitions are given below, but this means
that we will study the properties of ring homomorphisms from $\F_q[x]$ to
the skew polynomial ring $\L\ang{\tau}$, where $\tau$ satisfies the
commutation relation $\tau u = u^q \tau$ for $u$ in $\L$. Here, the
{\em rank} of such a morphism $\varphi$ is the degree in $\tau$ of
$\varphi(x)$.

Rank two Drinfeld modules enjoy remarkable similarities with elliptic curves: analogues exist of good reduction, complex
multiplication, etc. Based in part on these similarities, Drinfeld modules have recently
started being considered under the algorithmic viewpoint. For
instance, they have been proved to be unsuitable for usual forms of
public key cryptography~\cite{Scanlon01}; they have also been used to
design several polynomial factorization
algorithms~\cite{PaPo89,vanderHeiden04,Narayanan18,eschost2017arXiv171200669D};
recent work by Garai and Papikian discusses the computation of their
endomorphism rings~\cite{GaPa18}. Our goal is to study in detail the
complexity of computing the
characteristic polynomial of a rank two Drinfeld module over a finite field.

A fundamental object attached to an elliptic curve $E$ defined over a
finite field $\F_q$ is its Frobenius endomorphism $\pi:(x,y) \mapsto
(x^q,y^q)$; it is known to satisfy a degree-two relation with integer
coefficients called its {\em characteristic polynomial}. Much is known
about this polynomial: it takes the form $T^2 - h T + q$, for some
integer $h$ called the {\em trace} of $\pi$, with $\log_2(|h|) \le
\log_2(q)/2 + 1$ (this is known as the Hasse bound). In 1985, Schoof
famously designed the first polynomial-time algorithm for finding the
characteristic polynomial of such a curve~\cite{schoof85}. 

%The
%question has since then been the object of many further studies,
%ranging from improvements to Schoof's
%algorithm~\cite{Atkin92b,Elkies92} to algorithms for fields of small
%characteristic~\cite{Satoh00} or algorithms that simultaneously
%compute the trace of several reductions of a given curve defined over
%$\Q$~\cite{Harvey14}.

Our main objective is to investigate the complexity of a Drinfeld
analogue of this question. Given a rank two Drinfeld module over a
degree $n$ extension $\L$ of $\F_q$, one can define its Frobenius
endomorphism, and prove that it satisfies a degree-two relation $T^2 -
A T + B$, where $A$ and $B$ are now in $\F_q[x]$. As in the elliptic
case, $B$ is rather easy to determine, and of degree $n$. Hence, our
main question is the determination of the polynomial $A$, which is
known to have degree at most $n/2$ (note the parallel with the
elliptic case).

Contrary to the elliptic case, computing the characteristic polynomial
of a Drinfeld module is easily seen to be feasible in polynomial time:
it boils down to finding the $\Theta(n)$ coefficients of $a$, which
are known to satisfy certain linear relations. Gekeler detailed such
an algorithm in~\cite{GEKELE1991187}; we will briefly revisit it in
order to analyse its complexity, which turns out to be cubic in
$n$. Our main contributions in this paper are several new algorithms
with improved runtimes; we also present experimental results obtained
by an implementation based on NTL~\cite{shoup2001ntl}.

An implementation of Gekeler's algorithm was described
in~\cite{Jung00} and used to study the distribution of characteristic
polynomials of Drinfeld modules by computing several thousands of
them.

%%%%%%%%%%%%%%%%%%%%%%%%%%%%%%%%%%%%%%%%%%%%%%%%%%%%%%%%%%%%
%%%%%%%%%%%%%%%%%%%%%%%%%%%%%%%%%%%%%%%%%%%%%%%%%%%%%%%%%%%%
%%%%%%%%%%%%%%%%%%%%%%%%%%%%%%%%%%%%%%%%%%%%%%%%%%%%%%%%%%%%

\section{Preliminaries}

In this section, we introduce notation to be used throughout the
paper; we recall the basic definition of Drinfeld modules and state
precisely our main problem. For a general reference on these
questions, see for instance~\cite{Goss96}.

%%%%%%%%%%%%%%%%%%%%%%%%%%%%%%%%%%%%%%%%%%%%%%%%%%%%%%%%%%%%

\subsection{The Fields $\F_q$, $\K$ and $\L$}\label{ssec:not}

In all the paper, $\F_q$ is a given finite field, of order a prime
power $q$, and $\L \supset \F_q$ is another finite field of degree $n$
over $\F_q$. Explicitly, we assume that $\L$ is given as
$\L=\F_q[z]/\frakf$, for some monic irreducible $\frakf \in \F_q[z]$ of
degree $n$. When needed, we will denote by $\zeta \in \L$ the class $(z \bmod
\frakf)$.

In addition, we suppose that we are given a ring homomorphism $\gamma:
\F_q[x] \to \L$. The kernel $\ker(\gamma)$ of the mapping $\gamma:
\F_q[x] \to \L$ is a prime ideal of $\F_q[x]$ generated by a monic
irreducible polynomial $\frakp$, referred to as the
$\F_q[x]$-\textit{characteristic} of $\L$. Then, $\gamma$ induces an
embedding $\K := \F_q[x]/\frakp \to \L$; we will write $m := [\L :
  \K]$, so that $n = md$, with $d=\deg \frakp$. When needed, we will
denote by $\xi \in \K$ the class $(x \bmod \frakp)$.

Although it may not seem justified yet, we may draw a parallel with
this setting and that of elliptic curves over finite fields. As said
before, one should see $\F_q[x]$ playing here the role of $\Z$ in the
elliptic theory. The irreducible $\frakp$ is the analogue of a prime
integer $p$, so that the field $\K = \F_q[x]/\frakp$ is often thought
of as the ``prime field'', justifying the term ``characteristic'' for
$\frakp$. The field extension $\L$ will be the ``field of
definition'' of our Drinfeld modules.

We denote by $\pi: \L \to \L$ the $q$-power Frobenius $u \mapsto u^q$;
for $i \ge 0$, the $i$th iterate $\pi^i: \L \to \L$ is thus $u \mapsto
u^{q^i}$; for $i \le 0$, $\pi^i$ is the $i$th iterate of $\pi^{-1}$.

%%%%%%%%%%%%%%%%%%%%%%%%%%%%%%%%%%%%%%%%%%%%%%%%%%%%%%%%%%%%

\subsection{Skew Polynomials}

We write $\L\ang{\tau}$ for the ring of so-called {\em skew
  polynomials}
\begin{align}\label{def:skewpoly}
\L\ang{\tau} &= \{U=u_0 + u_1 \tau + \cdots + u_s \tau^s \ \mid \ s \in
\N, u_0,\dots,u_s \in \L\}.
\end{align}
This ring is endowed with the multiplication induced by the relation
$\tau u = u^q \tau$, for all $u$ in $\L$.  Elements of $\L\ang{\tau}$
are sometimes called linearized polynomials, since there exists an
isomorphism mapping $\L\ang{\tau}$ to polynomials of the form $u_0x +
u_1 x^q + \cdots + u_s x^{q^s}$, which form a ring for the operations
of addition and composition. 

A non-zero element $U$ of $\L\ang{\tau}$ admits a unique
representation as in~\eqref{def:skewpoly} with $u_s$ non-zero. Its
{\em degree} $\deg U$ is the integer $s$ (as usual, we set $\deg 0
=-\infty$).  The ring $\L\ang{\tau}$ admits a right Euclidean
division: given $U$ and $V$ in $\L\ang{\tau}$, with $V$ non-zero,
there exists a unique pair $(Q,R)$ in $\L\ang{\tau}^2$ such that $U =
QV +R$ and $\deg R < \deg V$.

There is a ring homomorphism
$\iota:\L\ang{\tau} \to {\rm End}_{\F_q}[\L]$ given by
\[\iota: u_0 + u_1\tau + \dots
+ u_s\tau^s \mapsto u_0 {\rm Id} + u_1 \pi + \dots+ u_s \pi^s, \] where
$\rm{Id}: \L \to \L$ is the identity operator and $\pi$ and its powers
are as defined above. This mapping allows us to interpret elements in
$ \L\ang{\tau}$ as $\F_q$-linear operators $\L \to \L$.

%%%%%%%%%%%%%%%%%%%%%%%%%%%%%%%%%%%%%%%%%%%%%%%%%%%%%%%%%%%%

\subsection{Drinfeld Modules}

Drinfeld modules can be defined in a quite general setting, involving
projective curves defined over $\F_q$; we will be concerned with
the following special case (where the projective curve in question
is simply $\P^1$).

\begin{definition}\label{def:Drinfeld}
  Let \spa $ \L$ and $\gamma$ be as above.  A rank $r$ {\em Drinfeld module}
  over $(\L,\gamma)$ is a ring homomorphism $\varphi: \F_q[x] \to
  \L\ang{\tau}$ such that
  \[\varphi(x) = \gamma(x) + u_1 \tau + \cdots + u_r\tau^r,\]
  with $u_1,\dots,u_r$ in $\L$ and $u_r$ non-zero.
\end{definition}
For $U$ in $\F_q[x]$, we will abide by the convention of writing
$\varphi_U$ in place of $\varphi(U)$. Since $\varphi$ is a ring
homomorphism, we have $\varphi_{UV} = \varphi_U \varphi_V$ and
$\varphi_{U+V} = \varphi_U+ \varphi_V$ for all $U,V$ in $\F_q[x]$;
hence, the Drinfeld module $\varphi$ is determined entirely by
$\varphi_x$; precisely, for $U$ in $\F_q[x]$, we have $\varphi_U =
U(\varphi_x)$.

We will restrict our considerations to rank two Drinfeld modules. In particular, we will use the now-standard convention of writing $\varphi_x = \gamma(x)
+ g \tau + \Delta \tau^2$. Hence, for a given $(\L,\gamma)$, we can
represent any rank two Drinfeld module over $(\L,\gamma)$ by the pair
$(g,\Delta) \in\L^2$.

\begin{example}
  Let $q = 5$, $\frakf = z^4 +4z^2+ 4z + 2$ and $\L = \F_5[z]/\frakf =
  \mathbb{F}_{625}$, so that $n=4$; we let $\zeta$ be the class of $z$ in
  $\L$.
  Let $\gamma: \F_5[x] \to \F_{625}$ be given by $x \mapsto \zeta$, so
  that $\frakp=\frakf$, $\K=\L=\F_{625}$ and $m=1$. We define the Drinfeld
  module $\varphi: \F_5[x] \to \F_{625}\ang{\tau}$ by $\varphi_x = \zeta + \tau + \tau^2$, so that $(g,\Delta)=(1,1)$.
\end{example}

%% For completeness, we briefly mention how our definitions can be
%% broadened. First, one may start from a smooth, projective, connected
%% curve $\mathcal{C}$ defined over $\F_q$ and replace the polynomial
%% ring $\F_q[x]$ by the subring $\A\subset \F_q(\mathcal{C})$ of
%% rational functions regular everywhere except possibly at some fixed
%% point $\infty \in \mathcal{C}$; the rest of the definition remains
%% unchanged. The definition given above has $\mathcal{C} = \P^1$, with
%% $\infty$ being its point at infinity and $\A=\F_q[x]$. Another
%% restriction that could be lifted is that $\L$ need not be a finite
%% field; any field with a mapping $\gamma$ as above would do (another
%% common choice is $\L=\F_q(x)$, which parallels the theory of elliptic
%% curves over $\Q$).

Suppose $\varphi$ is a rank two Drinfeld module over $(\L,\gamma)$. A
central element in $\L\ang{\tau}$ is called an {\em endormorphism} of
$\varphi$. Since $u^{q^n} = u$ for all $u$ in $\L$, $\tau^n$ is such
an endomorphism. The following key theorem~\cite[Cor. 3.4]{GEKELE1991187}
defines the main objects we wish to compute.
\begin{theorem}\label{charpoly}
  There is a polynomial $T^2 -AT + B \in \F_q[x][T]$ such that
  $\tau^n$ satisfies the  equation
  \begin{equation} \tau^{2n} - \varphi_A \tau^n + \varphi_B = 0,\end{equation}
  with $\deg A \le n/2$ and $\deg B=n$.
\end{theorem}
The polynomials $A$ and $B$ are respectively referred to as the
\textit{Frobenius trace} and \textit{Frobenius norm} of $\varphi$.
Note in particular the similarity with Hasse's theorem for elliptic
curves over finite fields regarding the respective ``sizes'' (degree,
here) of the Frobenius trace and norm.  The main goal of this paper is
then to find efficient algorithms to solve the following problem.
\begin{problem}\label{pb1}
  Given a rank two Drinfeld module $\varphi = (g,\Delta)$, compute its
  Frobenius trace $A$ and Frobenius norm $B$.
\end{problem}

\begin{example}
In the previous example, we have 
$A = 3x^2+x+3$ and $B =x^4+4x^2+4x+2$.
\end{example}

By composing $\varphi: \F_q[x]\to \L\ang{\tau}$ and
$\iota:\L\ang{\tau} \to {\rm End}_{\F_q}[\L]$ as defined in the
previous subsection, we obtain another ring homomorphism $\Phi:
\F_q[x]\to {\rm End}_{\F_q}[\L]$; we will use the same convention of
writing $\Phi_U=\Phi(U)$ for $U$ in $\F_q[x]$. Thus, we see that a
Drinfeld module equips $\L$ with a new structure as an
$\F_q[x]$-module, induced by the choice of $\Phi_x = \gamma(x) {\rm
  Id} + g \pi + \Delta \pi^2$, with $\pi:\L \to \L$ the $q$-power Frobenius map

Applying $\iota$ to the equality in Theorem~\ref{charpoly}, we obtain
that $\pi^{2n} + \Phi_A \pi^n + \Phi_B$ is the zero linear mapping $\L
\to \L$. Since $\pi^{n}$ is the identity map, and since we have
$\Phi_A = A(\Phi_x)$, $\Phi_B = B(\Phi_x)$, this implies that the
polynomial $1-A+B \in \F_q[x]$ cancels the $\F_q$-endormorphism
$\Phi_x$. Actually, more is true: $1-A+B$ is the characteristic
polynomial of this endomorphism~\cite[Th.~5.1]{GEKELE1991187}.  As it
turns out, finding the Frobenius norm $B$ is a rather easy task (see
Section~\ref{sec:prev}); as a result, Problem~\ref{pb1} can be reduced to
computing the characteristic polynomial of $\Phi_x$.

This shows in particular that finding $A$ and $B$ can be done in $(n
\log q)^{O(1)}$ bit operations (in all the paper, we will use a
boolean complexity model, which counts the bit complexity of all
operations on a standard RAM). The questions that interest us are to
make this cost estimate more precise, and to demonstrate algorithmic
improvements in practice, whenever possible. Our main results are as
follows.
\begin{theorem}\label{theo:main}
One can solve Problem~\ref{pb1}
\begin{itemize}
\item in Monte Carlo time $(n^{1.885} \log q + n \log^2 q)^{1+o(1)}$,
  if the minimal polynomial of $\Phi_x$ has degree $n$ (Section~\ref{sec:narayanan});
\item in time $(n^{2+\varepsilon} \log q + n \log^2 q)^{1+o(1)}$, for
  any $\varepsilon > 0$ (Section~\ref{sec:schoof});
\item in Monte Carlo time $(n^2 \log^2 q)^{1+o(1)}$ (Section~\ref{sec:mc}).
\end{itemize}
\end{theorem}
Section~\ref{sec:prev} reviews previous work; it shows that our
results are the best to date, except when $\K=\L$ (the ``prime field
case''), where a runtime $(n^{1.5+\varepsilon} \log q + n^{1 + \epsilon} \log^2
q)^{1+o(1)}$ is possible for any $\varepsilon > 0$~\cite{eschost2017arXiv171200669D}.
Section~\ref{sec:experiments} discusses the practical behavior of
these algorithms; in particular, it highlights that among all of them,
the Monte Carlo algorithm of Section~\ref{sec:mc} features the best
runtimes, except when $\K=\L$, where the above-mentioned result of~\cite{eschost2017arXiv171200669D}
is superior.

Input and output sizes are $\Theta(n \log q)$ bits, so the best we
could hope for is a runtime quasi-linear in $n \log q$; as the theorem
shows, we are rather far from this, since the best unconditional
results are quadratic in $n$. On the other hand, Problem~\ref{pb1} is
very similar to questions encountered when factoring polynomials over
finite fields, and it was not until the work of Kaltofen and
Shoup~\cite{KaSh98} that subquadratic factorization algorithms were
discovered.  We believe that
finding an algorithm of unconditional subquadratic time in $n$ for
Problem~\ref{pb1} is an interesting and challenging question.

The algorithm of Section~\ref{sec:schoof} was directly inspired by
Schoof's algorithm for elliptic curves. We believe this interaction has the potential to
yield further algorithms of interest, perhaps using other
``elliptic'' techniques, such as $p$-adic approaches~\cite{Satoh00} or
Harvey's amortization techniques~\cite{Harvey14}.

%%%%%%%%%%%%%%%%%%%%%%%%%%%%%%%%%%%%%%%%%%%%%%%%%%%%%%%%%%%%
%%%%%%%%%%%%%%%%%%%%%%%%%%%%%%%%%%%%%%%%%%%%%%%%%%%%%%%%%%%%
%%%%%%%%%%%%%%%%%%%%%%%%%%%%%%%%%%%%%%%%%%%%%%%%%%%%%%%%%%%%

\section{Algorithmic Background}

We now discuss the cost of operations in $\L$ and $\L\ang{\tau}$ with runtimes
given in bit operations. Notation ($\F_q,\L,\dots$) is as in~\ref{ssec:not}. To simplify cost
analyses, {\em we assume that $x^q \bmod \frakp$ is known}; we will
see below the cost of computing it once and for all, at the beginning
of our algorithms.

%%%%%%%%%%%%%%%%%%%%%%%%%%%%%%%%%%%%%%%%%%%%%%%%%%%%%%%%%%%%

\subsection{Polynomial and matrix arithmetic}\label{ssec:basicpoly}

\smallskip\noindent{\bf 3.1.1.}  Elements of $\L$ are written on 
the power basis $1,\zeta,\dots,\zeta^{n-1}$. On occasion, we use 
$\F_q$-linear forms $\L \to \F_q$; they are given on the dual 
basis, that is, by their values at $1,\zeta,\dots,\zeta^{n-1}$.

Using FFT-based multiplication,
polynomial multiplication, division and extended GCD in degree $n$,
and thus addition, multiplication and inversion in $\L$, can be done
in $(n\log q)^{1+o(1)}$ bit
operations~\cite{Gathen:2003:MCA:945759}. In particular, computing
$x^q \bmod \frakp$ by means of repeated squaring takes $(n\log^2
q)^{1+o(1)}$ bit operations.

\smallskip\noindent{\bf 3.1.2.}  We let $\omega$ be such that over any
ring, square matrix multiplication in size $s$ can be done in
$O(s^\omega)$ ring operations; the best known value to date is $\omega
\le 2.373$~\cite{CoWi90,LeGall14}. Using block techniques,
multiplication in sizes $(s,t) \times (t,u)$ takes $O(stu
\min(s,t,u)^{\omega-3})$ ring operations. For matrices over
$\F_q$, this is $(stu \min(s,t,u)^{\omega-3} \log q)^{1+o(1)}$ bit
operations; over $\L$, it becomes $(stu \min(s,t,u)^{\omega-3}
n\log q)^{1+o(1)}$.

We could sharpen our results using the so-called
exponent $\omega_2$ of rectangular matrix multiplication in size
$(s,s) \times (s,s^2)$. We can of course take $\omega_2\le\omega+1 \le
3.373$, but the better result $\omega_2 \le 3.252$ is
known~\cite{LeUr18}. We will not use these refinments in this paper.

\smallskip\noindent{\bf 3.1.3.} Of particular interest is an operation
called {\em modular composition}, which maps $(F,G,H) \in \F_q[x]^3$
to $F(G) \bmod H$, with $\deg H=n$ and $\deg F,\deg G < n$.  Let
$\theta \in [1,2]$ be such that this can be done in $(n^\theta
\log q)^{1+o(1)}$ bit operations for inputs of degree $O(n)$.

Modular composition is linear in $F$; we also require that its
transpose map can be computed in the same runtime $(n^\theta
\log q)^{1+o(1)}$. In an algebraic model, counting $\F_q$-operations
at unit cost, the {\em transposition principle}~\cite{KaKiBs88}
guarantees this, but this is not necessarily the case in our bit
model, hence our extra requirement.

For long, the best known value for $\theta$ was Brent and Kung's
$\theta = (\omega+1)/2$~\cite{BrKu78}. A major result by Kedlaya and
Umans proves that we can actually take $\theta = 1+\varepsilon$, for
any $\varepsilon >
0$~\cite{Kedlaya:2011:FPF:2340436.2340448}. In practical
terms, we are not aware of an implementation of Kedlaya and Umans'
algorithm that would be competitive: for
practical purposes, $\theta$ is either $(\omega+1)/2$ (for deterministic approaches) or $(\omega+2)/3$, and $\omega$ itself is
either 3 or Strassen's $\log_2 7 \simeq 2.81$. 
%We discuss some
%consequences of this remark in the next paragraph and in {\bf 3.2.4}.

\smallskip\noindent{\bf 3.1.4.}  A useful application of modular
composition is the application of any power of the Frobenius map
$\pi$: given $x^q \bmod \frakp$, for any $\alpha$ in $\L$ and $i \in
\{-(n-1),\dots,n-1\}$, we can compute $\pi^i(\alpha)$ for $O(\log n)$
modular compositions, that is, in $(n^\theta \log q)^{1+o(1)}$ bit
operations. See for instance~\cite[Algorithm~5.2]{vonzurGathen1992}
or Section~2.2 in~\cite{eschost2017arXiv171200669D}.

For small values of $i$, say $i=O(1)$, the computation of
$\pi^i(\alpha)$ can also be done by repeated squaring, in $(n \log^2
q)^{1+o(1)}$ bit operations.  Since for all implementations we are
aware of, $\theta=(\omega+1)/2$, this approach may be preferred for
moderate values of $\log q$ (this also applies to the operation in the
next paragraph).

\smallskip\noindent{\bf 3.1.5.} The previous item implies that if
$\varphi=(g,\Delta)$ is a rank two Drinfeld module over $(\L,\gamma)$,
given $\alpha$ in $\L$, we can compute $\Phi_x(\alpha) = \gamma(x)
\alpha + g \pi(\alpha) + \Delta \pi^2(\alpha)$ in time $(n^{\theta}
\log q)^{1+o(1)}$.  Because of our requirements on $\theta$, the same
holds for the {\em transpose} of $\Phi_x$: given an $\F_q$-linear form
$\ell:\L\to \F_q$, with the convention of {\bf 3.1.1}, we can compute
the linear form $\Phi_x^\perp(\ell): \alpha\mapsto
\ell(\Phi_x(\alpha))$ for the same cost.

%%%%%%%%%%%%%%%%%%%%%%%%%%%%%%%%%%%%%%%%%%%%%%%%%%%%%%%%%%%%

\subsection{Skew Polynomial Arithmetic}

\smallskip\noindent{\bf 3.2.1.}  We continue with skew polynomial
multiplication. This is an intricate question, with several algorithms
co-existing; which one is the most efficient depends on the input
degree. We will be concerned with multiplication in degree $k$, for
some $k \ll n$; in this case, the best algorithm to date is
from~\cite[Th.~7]{PUCHINGER2017b}. For any $k$, that algorithm uses
$O(k^{(\omega+1)/2})$ operations $+,\times$ in $\L$, together with
$O(k^{3/2})$ applications of powers of the Frobenius, for a total of
$(k^{(\omega+1)/2} n^\theta \log q)^{1+o(1)}$ bit operations.  For
higher degrees $k$, the algorithms in~\cite{CaLe17} have a better
runtime.

\smallskip\noindent{\bf 3.2.2.} Our next question is to compute
$\varphi_{x^k}$, for some $k \ge 0$; this polynomial has
$\Theta(k)$ coefficients in $\L$, so it uses $\Theta(k n \log q)$
bits. Since $\varphi_{x^{2k}}=\varphi_{x^{k}}\varphi_{x^{k}}$
and $\varphi_{x^{2k+1}} = \varphi_x\varphi_{x^{2k}}$, we can obtain
$\varphi_{x^k}$ from $\varphi_{x^{\floor*{k/2}}}$ using
$ (k^{(\omega+1)/2} n^\theta \log q)^{1+o(1)}$ bit operations.
The cumulated time to obtain $\varphi_{x^k}$ from $\varphi_x$ admits
the same upper bound.

\smallskip\noindent{\bf 3.2.3.} We consider now the cost of computing
$\varphi_C$, for some $C$ in $\F_q[x]$. To this end, we adapt the
divide-and-conquer algorithm
of~\cite[Ch.~9]{Gathen:2003:MCA:945759}, which applies to
commutative polynomials.
\begin{enumerate}
\item First, choose a power of two $k$ such that $k/2 \le \deg C <
  k$. We compute $\varphi_{x^i}$, for all $i$ powers of two up to
  $k/2$; using {\bf 3.2.2}, the cost is $(k^{(\omega+1)/2}
  n^\theta\log q)^{1+o(1)}$.
\item Write $C = C_0 + x^{k/2} C_1$, with $\deg C_0,\deg C_1 <
  k/2$. Compute recursively $\varphi_{C_0}$ and $\varphi_{C_1}$, and
  return $\varphi_C = \varphi_{C_0} + \varphi_{x^{k/2}}
  \varphi_{C_1}$.
\end{enumerate}
The cumulated cost of all recursive calls is
$(k^{(\omega+1)/2} n^\theta\log q)^{1+o(1)}$, which is $(\deg
(C)^{(\omega+1)/2} n^\theta\log q)^{1+o(1)}$.

\smallskip\noindent{\bf 3.2.4.}
Next, we analyze the cost of computing 
$\varphi_1,\varphi_x,\dots,\varphi_{x^k}$, for some $k \ge 0$. In
this, we essentially follow a procedure used by
Gekeler~\cite[Sec.~3]{frobdist}, although the cost analysis is not
done in that reference. These polynomials satisfy the following
recurrence:
\begin{align*}
 \varphi_{x^{i+1}} = \varphi_x \varphi_{x^i} & = (\gamma(x) + g\tau + \Delta \tau^2) \varphi_{x^i}.
\end{align*}
For $i \ge 0$, write 
\[\varphi_{x^i}  = \sum_{0 \le j \leq 2i} f_{i,j} \tau^{j},\]
for some coefficients $f_{i,j} \in \L$ to be determined. We obtain
\begin{align*}
 \sum_{0 \le j \leq 2i} f_{i,j} \tau^{j} = \sum_{0 \le j \leq 2i} \gamma(x) f_{i,j} \tau^{j} + \sum_{j \leq 2i} g f_{i,j}^q \tau^{j+1} + \sum_{j \leq 2i} \Delta f_{i,j}^{q^2} \tau^{j+2},
\end{align*}
so the $f_{i,j}$ satisfy the recurrence
\[ f_{i+1,j} = \gamma(x) f_{i,j} + g f_{i,j-1}^q + \Delta f_{i,j-2}^{q^2}\]
with known initial conditions $f_{0,0} = 1$, $f_{1,0} = \gamma(x)$,
$f_{1,1} = g$, and $f_{1,2} = \Delta$. Evaluating one instance of the
recurrence involves $O(1)$ multiplications / additions in $\L$ and
applications of the Frobenius map $\pi$, for $(n^\theta \log
q)^{1+o(1)}$ bit operations.  Given $\varphi_{x^i}$, there are
$\Theta(i)$ choices of $j$, so the overall cost to obtain
$\varphi_1,\varphi_x,\dots,\varphi_{x^k}$ is $(k^2 n^\theta \log
q)^{1+o(1)}$ bit operations.
In particular, taking $\theta = 1+\varepsilon$, we see that the
runtime here is essentially linear in the output size, which is
$\Theta(k^2 n \log q)$ bits; this was not the case for the algorithms 
in {\bf 3.2.1} - {\bf 3.2.2} - {\bf 3.2.3}.

However, in {\bf 3.1.3}, we pointed out that in practice, Brent and Kung's
modular composition algorithm is widely used, with
$\theta=(\omega+1)/2$. In this case, for moderate values of $\log q$,
one may use the straightforward repeated squaring method to apply the
Frobenius map; this leads to a runtime of $(k^2 n \log^2 q)^{1+o(1)}$
bit operations, which may be acceptable in practice. This also
applies to Proposition~\ref{prop:gek} below, and underlies the design
of the algorithm in Section~\ref{sec:mc}.

\smallskip\noindent{\bf 3.2.5.}  We deduce from this an algorithm for
inverting $\varphi$.  Given
$\varphi_C = \sum_{0 \le i \le 2k} \alpha_i \tau^i$, we want to
recover $C =\sum_{0 \le i \le k} c_ix^i $ in $\F_q[x]$. Writing the
expansion
\begin{align*}\label{eq:phic}
\varphi_C &= \sum_{0 \le i \le k} c_i \sum_{0 \le j \le 2i} f_{i,j} \tau^j = \sum_{0 \le j \le 2k}  \left (\sum_{\lfloor j/2\rfloor \le i \le k} c_i f_{i,j} \right) \tau^j.
\end{align*}
gives us $2k+1$ equations in $k+1$
unknowns. Keeping only those equations corresponding to even degree
coefficients leaves the following upper triangular system of $k + 1$
equations over $\L$,
\begin{equation}
\begin{bmatrix} f_{0,0} & f_{1,0} & \ldots & f_{k, 0} \\
                 0      & f_{1,2} & \ldots & f_{k, 2}  \\
                 
                 \vdots  & \vdots  &  \ddots      & \vdots                       \\
                 0  & 0 & \ldots & f_{k, 2k}
\end{bmatrix}
\begin{bmatrix}
  c_0 \\ c_1 \\ \vdots \\ c_k
\end{bmatrix} = \begin{bmatrix} \alpha_{0} \\ \alpha_{2} \\ \vdots \\ \alpha_{2k} \end{bmatrix}.
\end{equation}
Its diagonal entries are of the form $f_{i,2i}$; these are the
coefficients of the leading terms of $\varphi_{x^i}$, so that for all
$i$, $f_{i,2i} = \Delta^{e_i}$ for some exponent $e_i$. In particular,
since $\Delta \neq 0$, the diagonal terms are non-zero, which allows
us to find $c_0,\dots,c_{k}$. Once we know all $f_{i,j}$'s, the cost
for solving the system is $O(k^2)$ operations in $\L$, so the total is
$ (k^2 n^{\theta} \log q)^{1+o(1)}$ bit operations.

\smallskip\noindent{\bf 3.2.6.}  Finally, we give an algorithm to
evaluate a degree $k$ skew polynomial $U$ at $\mu$ elements
$\alpha_1,\dots,\alpha_\mu$ in $\L$. This algorithm will be used
only in Section~\ref{sec:narayanan}, so it can be skipped on first reading.

In the case of commutative
polynomials, one can compute all
$U(\alpha_i)$ faster than by successive evaluation of $U$ at
$\alpha_1$, $\alpha_2$, \dots;
see~\cite[Ch.~10]{Gathen:2003:MCA:945759}. The same holds for 
skew polynomial evaluation: in~\cite[Th.~15]{PUCHINGER2017b},
Puchinger and Wachter-Zeh gave an algorithm that uses
$O(k^{\max(\log_2 3, \omega_2/2)}\log k)$ operations in $\L$
(including Frobenius-powers applications) in the case $\mu=k$, where
$\omega_2 \le \omega+1$ as in {\bf 3.2.1}.

We propose a baby-step / giant-step procedure that applies
to any $\mu$ and $k$ (but the cost analysis depends on whether
$\mu \le \sqrt{k}$ or not). Suppose without loss of
generality that our input polynomial $U = u_0 + \cdots +
u_{k-1}\tau^{k-1}$ has degree less than $k$, for some perfect square
$k$, and let $s=\sqrt{k}$.
\begin{enumerate}
\item Commute powers of $\tau$ with the coefficients of $U$ to rewrite it as $U = U^*_{0} + \tau^s U^*_{1} + \cdots + \tau^{s(s-1)} U^*_{s-1}$.
  with all $U^*_i$ in $\L\ang{\tau}$ of degree less than $s$.
This is $O(k)$  applications of Frobenius powers in $\L$.
\item Compute $\alpha_{i,j}:=\pi^i(\alpha_j)$, for
  $i=0,\dots,s-1$ and $j=1,\dots,\mu$; this is $O(s \mu)$
applications of Frobenius powers.
\item For $i< s$, let $u^*_{i,0},\dots,u^*_{i,s-1}$ be 
the coefficients of $U^*_i$. Compute the matrix 
$(s,s) \times (s,\mu)$ product
\begin{align*}
\left [ \begin{matrix}
u^*_{0,0} & \cdots & u^*_{0,s-1} \\
\vdots && \vdots \\
u^*_{s-1,0} & \cdots & u^*_{s-1,s-1} 
  \end{matrix} \right ]
\left [ \begin{matrix}
\alpha_{0,1} & \cdots & \alpha_{0,\mu} \\
\vdots && \vdots \\
\alpha_{s-1,1} & \cdots & \alpha_{s-1,\mu} 
  \end{matrix} \right ],
\end{align*}
whose entries are the values $\beta_{i,j}:=U^*_{i}(\alpha_j)$.  When
we apply this result, we will have $\mu \le s = \sqrt{k}$, so the cost is $O(k
\mu^{\omega-2})$ operations in~$\L$, by~{\bf 3.1.2}. For completeness,
we mention that if $\mu \ge \sqrt{k}$, the cost is
$O(k^{(\omega-1)/2} \mu)$  operations in~$\L$.

\item Using Horner's scheme, for $j=1,\dots,\mu$, recover $U(\alpha_j)$ using 
$U(\alpha_j) = \beta_{0,j} + \tau^s( \beta_{1,j} + \tau^s(\beta_{2,j}
  + \cdots ))$. The total is another $O(s \mu)$ operations
in $\L$, including Frobenius powers.
\end{enumerate}
When $\mu \le \sqrt{k}$, the cost is $(k \mu^{\omega-2}
n^\theta \log q)^{1+ o(1)}$ bit operations. If we take $\mu\ge
\sqrt{k}$, the cost  becomes $(k^{(\omega-1)/2} \mu n^\theta \log
q)^{1+ o(1)}$.

%%%%%%%%%%%%%%%%%%%%%%%%%%%%%%%%%%%%%%%%%%%%%%%%%%%%%%%%%%%%
%%%%%%%%%%%%%%%%%%%%%%%%%%%%%%%%%%%%%%%%%%%%%%%%%%%%%%%%%%%%
%%%%%%%%%%%%%%%%%%%%%%%%%%%%%%%%%%%%%%%%%%%%%%%%%%%%%%%%%%%%

\section{Previous Work on Problem~\ref{pb1}}\label{sec:prev}

Next, we briefly review existing algorithms for solving
Problem~\ref{pb1}, and comment on their runtime. Notation are still
from Section~\ref{ssec:not}.

%%%%%%%%%%%%%%%%%%%%%%%%%%%%%%%%%%%%%%%%%%%%%%%%%%%%%%%%%%%%

\subsection{Gekeler's Algorithm}\label{ssec:gek}

As with elliptic curves, determining the Frobenius norm $B$
of Theorem~\ref{charpoly} is simply
done using the following result from~\cite[Th.~2.11]{frobdist}.

\begin{proposition}\label{frobnorm}
Let $N_{\L/\F_q}$ be the norm $\L \to \F_q$. The Frobenius norm $B$ of a rank two Drinfeld module $\varphi=(g,\Delta)$ 
  over $(\L,\gamma)$ is
  \[B = (-1)^n N_{\L/\F_q}(\Delta)^{-1}\frakp^m.\]
\end{proposition}
In particular, $B$ can be computed in $(n \log q)^{1+o(1)}$ bit
operations. Indeed, $\frakp^m$ is a degree $n$ polynomial, and we can
compute it in the prescribed time by repeated squaring. Moreover
$N_{\L/\F_q}(\Delta) =
\textnormal{resultant}(\frakf,\Delta)$~\cite{Pohst:1989:AAN:76692},
so we can compute it in the same time~\cite{Gathen:2003:MCA:945759}.

Gekeler also gave in \cite[Sec.~3]{frobdist} an algorithm that
determines the Frobenius trace $A$ by solving a linear system for the
coefficients of $A$. The key subroutines used in this algorithm were
described in the previous section, and imply the following result (the
cost analysis is not provided in the original paper).
\begin{proposition}\label{prop:gek}
  One can solve Problem~\ref{pb1} using $(n^{\theta+2}\log
  q + n \log^2 q)^{1+o(1)}$ bit operations.
\end{proposition}
\begin{proof}
The algorithm is as follows.
\begin{enumerate}
\item We compute $x^q\bmod \frakp$ with $(n \log^2
  q)^{1+o(1)}$ bit operations.
\item Find $\varphi_1,\dots,\varphi_{x^n}$ in
$(n^{\theta+2} \log q)^{1+o(1)}$ bit operations ({\bf 3.2.4}).
\item Compute $B$ and deduce $\varphi_B$; this takes comparatively
  negligible time (see above and {\bf 3.2.3}) and gives us
  $\varphi_A$, since Theorem~\ref{charpoly} implies that $\tau^n
  \varphi_A = \tau^{2n} + \varphi_B.$
\item Recover $A$ in $(n^{\theta+2} \log q)^{1+o(1)}$
bit operations by {\bf 3.2.5}. \qedhere
\end{enumerate}
\end{proof}
The cost of this procedure is at least cubic in $n$, due to the need
to compute the $\Theta(n^2)$ coefficients $f_{i,j}$ of
$\varphi_1,\dots,\varphi_{x^n}$ in $\L$.

%%%%%%%%%%%%%%%%%%%%%%%%%%%%%%%%%%%%%%%%%%%%%%%%%%%%%%%%%%%%

\subsection{The Case $\L = \K$}

The case where $\L = \K$, that is, when $\gamma: \F_q[x] \to \L$ is
onto, allows for some faster algorithms, based on two observations:
we can recover $A$ from its image $\gamma(A)$ in this case (since
 $\deg A \le \floor*{n/2}$), and $\gamma(A)$ can be easily
derived from the {\em Hasse Invariant} of $\varphi$, which is the
coefficient of $\tau^{n}$ in $\varphi_{\frakp}={\frakp}(\varphi_x)$.

From this, Hsia and Yu~\cite{HsYu00} and Garai and
Papikian~\cite{GaPa18} sketched algorithms that compute $A$.  When
$\varphi_{\frakp}$ is computed in a direct manner, they take
$\Theta(n^2)$ additions, multiplications and Frobenius applications in
$\L$, so $\Omega(n^3)$ bit operations.

Gekeler~\cite[Prop.~3.7]{frobdist} gave an algorithm inspired by an
analogy with the elliptic case, where the Hasse invariant can be
computed as a suitable term in a recurrent sequence (with non-constant
coefficients). A direct application of this result does not improve on
the runtime above. However, using techniques inspired by both the
elliptic case~\cite{BoGaSc07} and the polynomial factorization
algorithm of~\cite{KaSh98}, it was shown
in~\cite{eschost2017arXiv171200669D} how to reduce the cost to
$(n^{\theta+1/2} \log q + n \log^2 q)^{1+o(1)}$ bit operations, which
is subquadratic in $n$.

%%%%%%%%%%%%%%%%%%%%%%%%%%%%%%%%%%%%%%%%%%%%%%%%%%%%%%%%%%%%
%%%%%%%%%%%%%%%%%%%%%%%%%%%%%%%%%%%%%%%%%%%%%%%%%%%%%%%%%%%%
%%%%%%%%%%%%%%%%%%%%%%%%%%%%%%%%%%%%%%%%%%%%%%%%%%%%%%%%%%%%

\section{On Narayanan's Algorithm}\label{sec:narayanan}

In~\cite[Sec.~3.1]{Narayanan18}, Narayanan gives the sketch of a Monte
Carlo algorithm to solve Problem~\ref{pb1} for odd $q$, which applies to those
Drinfeld modules $(g,\Delta)$ for which the minimal polynomial $\Gamma$ of
$\Phi_x=\gamma(x){\rm Id} + g \pi + \Delta \pi^2$ has degree $n$. In
this case, it must coincide with the characteristic polynomial of
$\Phi_x$, which we saw is equal to $1-A+B$ (this assumption on
$\Gamma$ holds for more than half of elements of the parameter
domain~\cite[Th.~3.6]{Narayanan18}). Since $B$ is easy to compute,
knowing $1-A+B$ gives us $A$ readily.

Narayanan's algorithm computes the minimal polynomial
$\Gamma_{\ell,\alpha}$ of a sequence of the form $(r_k)_{k \ge 0} =
(\ell(\Phi_x^k(\alpha))_{k\ge 0} \in {\F_q}^\N$, for a random $\F_q$-linear
map $\ell:\L\to\F_q$ and a random $\alpha \in \L$. Using Wiedemann's
analysis~\cite{Wiedemann:1986:SSL:13738.13744}, one can bound below
the fraction of $\ell$ and $\alpha$ for which
$\Gamma_{\ell,\alpha}=\Gamma$.
The bottleneck of this algorithm is the computation of sufficiently
many elements of the above sequence: the first $2n$ terms are needed,
after which applying Berlekamp-Massey's algorithm gives us
$\Gamma_{\ell,\alpha}$. To compute $(r_k)_{0 \le k < 2n}$, Narayanan states
that we can adapt the {\em automorphism projection} algorithm of
Kaltofen and Shoup~\cite{KaSh98} and enjoy its subquadratic
complexity. Indeed, Kaltofen and Shoup's algorithm computes terms in a similar
sequence, namely $\ell(\pi^k(\alpha))_{k\ge 0}$, where $\pi$ is 
the Frobenius map. However, that
algorithm actively uses the fact that $\pi$ is a field automorphism, whereas
$\Phi_x$ is not. Hence, whether a direct adaptation of Kaltofen
and Shoup's algorithm is possible remains unclear to us.

We propose an alternative Monte Carlo algorithm, which
establishes the first point in Theorem~\ref{theo:main}; it is
inspired by Coppersmith's block Wiedemann
algorithm~\cite{Coppersmith94}. 

The sequence $(\ell(\Phi_x^k(\alpha))_{k\ge 0}$ used in Wiedemann's
algorithm is linearly recurrent, so that its generating series is
rational, with $\Gamma$ as denominator for generic choices of $\ell$
and $\alpha$. In Coppersmith's block version, we consider a sequence
of $\mu \times \mu$ matrices $({\bm R}_k)_{k \ge 0}$ over $\F_q$
instead, for some given parameter $\mu$. These matrices are defined by
choosing $\mu$ many $\F_q$-linear mappings $\L\to\F_q$, say ${\bm
  \ell}=(\ell_1,\dots,\ell_\mu)$, and $\mu$ elements ${\bm
  \alpha}=(\alpha_1,\dots,\alpha_\mu)$ in $\L$.  They define sequences
$(r_{i,j,k})_{k \ge 0} := (\ell_i(\Phi_x^k(\alpha_j)))_{k \ge 0}$,
which form the entries of a sequence of $\mu \times \mu$ matrices
$({\bm R}_k)_{k \ge 0}$.  The generating series $\sum_{k \ge 0} {\bm
  R}_k / z^{k+1}$ can be written as ${\bm Q}^{-1}{\bm N}$, for some
${\bm Q}$ and ${\bm N}$ in $\F_q[z]^{\mu \times \mu}$.  For generic
choices of ${\bm \ell}$ and ${\bm \alpha}$, ${\bm Q}$ has degree at
most $\lceil n/\mu \rceil$ and can be computed in $(\mu^{\omega-1} n \log
q)^{1+o(1)}$ bit operations from $({\bm R}_k)_{k \le 2n/\mu}$, using
the PM basis algorithm of~\cite{GiJeVi03}. Finally, we will see that
we can deduce the minimal polynomial $\Gamma$ from the determinant of
${\bm Q}$.

Thus, Coppersmith's algorithm requires fewer values of the matrix
sequence than Wiedemann's (roughly $2n/\mu$ instead of $2n$). As we
will see, the multipoint evaluation algorithm in {\bf 3.2.6} makes it
possible to compute all required matrices in subquadratic time.
The overview of the algorithm is thus the following.
\vspace{-0.075cm}
\begin{enumerate}
\item Fix $\mu=\lfloor n^{b} \rfloor$, for some exponent $b$ to be
  determined later; choose $\mu$ many $\F_q$-linear mappings $\L\to\F_q$,
 ${\bm \ell}=(\ell_1,\dots,\ell_\mu)$, and $\mu$ elements ${\bm
    \alpha}=(\alpha_1,\dots,\alpha_\mu)$ in $\L$.  
\item Compute $({\bm R}_k)_{0 \le k \le 2n/\mu}$, for ${\bm R}_k$ as
  defined above. We will discuss the cost of this operation below.
\item Compute ${\bm Q}$; this takes $(\mu^{\omega-1} n \log q)^{1+o(1)}$ bit operations.
\item Compute the determinant $\Gamma^*$ of ${\bm Q}$.  The cost of
  this step is another $(\mu^{\omega-1} n \log q)^{1+o(1)}$ bit
  operations~\cite{LaNeZh17}.  By~\cite[Th.~2.12]{KaVi04}, $\Gamma^*$
  divides the characteristic polynomial of $\Phi_x$, which we assume
  coincides with $\Gamma$. For generic $\ell_1$ and
  $\alpha_1$, the minimal polynomial of $(r_{1,1,k})_{k \ge 0}$ is $\Gamma$. If this is the case, since $\Gamma^*$ cancels
  that sequence, $\Gamma$ divides $\Gamma^*$, so that $\Gamma =
  \Gamma^*$.

%% S = Q^{-1}N so QS = N so det(Q) S = Adj.N so det(Q). each entry of S is a polynomial
\end{enumerate}

Regarding the probabilistic aspects, combining the last paragraphs
of~\cite[Sec.~2.1]{KaVi04} (that deal with the properties of ${\bm
  Q}$) and the analysis
in~\cite{Kaltofen:1991:PEP:113379.113396,Kaltofen-saun:1991:WMS:646027.676885}
(for Step 4) shows that there is a non-zero polynomial $D$ in
$\F_q[{\bm L}_1,\dots,{\bm L}_\mu,{\bm A}_1,\dots,{\bm A}_\mu]$, where
each boldface symbol is a vector of $n$ indeterminates, such that
$\deg D \le 4n$, and such that if
$D({\ell}_1,\dots,{\ell}_\mu,{\alpha}_1,\dots,{\alpha}_\mu)\ne 0$, all
properties above hold. By the DeMillo-Lipton-Zippel-Schwartz lemma,
the probability of failure is thus at most $4n/q$. If $q < 4n$, we may
have to choose the coefficients of ${\bm \ell}$ and ${\bm \alpha}$ in
an extension of $\F_q$ of degree $O(\log n)$; this affects the runtime
only with respect to logarithmic factors.

It remains to explain how to compute the required matrix values $({\bm
  R}_k)_{k \le 2n/\mu}$ at step (2). This is done by adapting the
baby-steps / giant steps techniques of~\cite[Algorithm {\bf
    AP}]{KaSh98} to the context of the block-Wiedemann algorithm, and
leveraging multipoint evaluation.  Let $K:=\floor*{ (n/2\mu)^{c}}$,
for another constant $c$ to be determined, and $K':=\lceil n/(2K\mu )
\rceil$; remark that $K' \mu\le n$.  For our final choices of
parameters, we will also have the inequalities $K' \le K$, $\mu \le
\sqrt{K}$.
\begin{enumerate}
\item[(2.1)] For $i \le \mu$ and $u < K$, compute the linear mapping
  $\ell_{i,u}:= {\Phi_x^\perp}^u(\ell_i)$, so that $\ell_{i,u}(\beta)
  = \ell_i( \Phi_x^u(\beta) )$ for $\beta$ in $\L$. By {\bf 3.1.5},
  this takes $ (K \mu n^{\theta} \log q)^{1+o(1)}$ bit
  operations.

\item[(2.2)] Compute $\varphi_{x^K} \in \L\ang{\tau}$; 
  this takes $(K^{(\omega+1)/2}n^\theta \log q)^{1+o(1)}$ bit
  operations, by {\bf 3.2.2}.

\item[(2.3)] For $j \le \mu$ and $v < K'$, compute $\alpha_{j,v} :=
  \Phi_x^{Kv}(\alpha_j)$, so that we have
  $\ell_{i,u}(\alpha_{j,v}) = \ell_i( \Phi_x^{u+Kv}(\alpha_j) )$ for
  all $i,j,u,v$.

Starting from $(\alpha_{1,v},\dots,\alpha_{\mu,v})$, the application
of $\varphi_{x^K}$ gives
$(\alpha_{1,v+1},\dots,\alpha_{\mu,v+1})$. This takes $ (K
\mu^{\omega-2} n^{\theta} \log q)^{1+o(1)}$ bit operations per index
$v$ (by {\bf 3.2.6}), so that the total cost is $(\mu^{\omega-3} n^{\theta+1}
\log q)^{1+o(1)}$ (note that $\mu \le \sqrt{K}$).

\item[(2.4)] Multiply the $(K\mu,n) \times (n,K'\mu)$ matrices with
  entries the coefficients of $(\ell_{1,0},\dots,\ell_{\mu,K-1})$,
  resp.\, $(\alpha_{1,0},\dots,\alpha_{\mu,K'-1})$, to obtain all
  needed values $r_{i,j,u+Kv}$.  The inequalities above imply that the
  smallest dimension is $K'\mu$ so by {\bf 3.1.2}, the cost is $(K^{3-\omega}\mu n^{\omega-1} \log q)^{1 + o(1)}$ bit operations.
\end{enumerate}
We know that we can take $\theta=1+\varepsilon$, for any $\varepsilon
> 0$. To find $b$ and $c$ that minimize the overall exponent in $n$, we can
thus replace $\theta$ by $1$ and disregard the exponent $1+o(1)$ and
the terms depending on $\log q$; we will then round up the final
result.  The relevant terms are $\{K\mu n,
K^{(\omega+1)/2}n, \mu^{\omega-3} n^2,  K^{3-\omega}\mu n^{\omega-1},
\mu^{\omega-1} n\}$.  For $\omega = 2.373$, taking $b = 0.183$ and
$c=0.642$, all inequalities we needed are satisfied and the runtime is
$(n^{1.885} \log q)^{1 + o(1)}$ bit operations. 

Taking into account the initial cost of computing $x^q$ in $\L$, this
proves the first point in our main theorem. It should however be
obvious from the presentation of the algorithm that we make no claims
as to its practical behavior (for instance, parameters $b,c$ were
determined using an exponent $2.373$ for matrix multiplication,
which is currently unrealistic in practice).

%%%%%%%%%%%%%%%%%%%%%%%%%%%%%%%%%%%%%%%%%%%%%%%%%%%%%%%%%%%%
%%%%%%%%%%%%%%%%%%%%%%%%%%%%%%%%%%%%%%%%%%%%%%%%%%%%%%%%%%%%
%%%%%%%%%%%%%%%%%%%%%%%%%%%%%%%%%%%%%%%%%%%%%%%%%%%%%%%%%%%%

\section{A Deterministic Algorithm}\label{sec:schoof}

We present next an alternative approach inspired by Schoof's algorithm
for elliptic curves, establishing the second item in our main theorem:
{\em we can solve Problem~\ref{pb1} in time $(n^{2+\varepsilon} \log q + n
\log^2 q)^{1+o(1)}$, for any $\varepsilon > 0$.} As before, we assume 
that we know $x^q \bmod \frakp$.

\smallskip\noindent{\bf 6.1.} We first compute the Frobenius norm
$B$. The idea of the algorithm is then to compute $A_i:=A \bmod E_i$,
for some pairwise distinct irreducible polynomials $E_1,\dots,E_s$ in
$\F_q[x]$ and recover $A$ by Chinese remaindering.  Thus, we
need $\deg(E_1 \cdots E_s) > n/2$, and we will also impose that $\deg
E_i \in O(\log n)$ for all $i$. First, we show that we can find such
 $E_i$'s in $(n^2 \log q)^{1+o(1)}$ bit operations.

If $q > n/2$, it is enough to take $E_i = x-e_i$, for pairwise
distinct elements $e_i$ in $\F_q$; enumerating $n/2+1$ elements of
$\F_q$ takes $(n \log q)^{1+o(1)}$ bit operations.

Otherwise, let $t= \lceil \log_q (n+1)\rceil$. The sum of the degrees
of the monic irreducible polynomials of degree $t$ over $\F_q$ is at
least $(1/2) q^t$, which is greater than $n/2$. Thus, we test all monic
polynomials of degree $t$ for irreducibility. There are $q^t < q(n+1)
\le n^2$ such polynomials (note that here $q \le n/2$) and each
irreducibility test takes $\log^{O(1)} n$ bit
operations~\cite{vonzurGathen1992} (a term $\log q$ usually appears in
such runtime estimates, but here $\log q$ is in $O(\log n)$).

Without loss of generality, we assume that no polynomial $E_i$ is such
that $E_i(\gamma(x))=0$ (recall that $\gamma$ is the structural
homomorphism $\F_q[x] \to \L$). Only one irreducible polynomial may
satisfy this equality, so we discard it and find a replacement if
needed.

\smallskip\noindent{\bf 6.2.} Let $F \in\L\ang{\tau}$ be of degree
$\delta$ and $\L\ang{\tau}_{\delta}$ be the set of all elements in
$\L\ang{\tau}$ of degree less than $\delta$. Our main algorithm will
rely on the following operation: define the operator $\text{\sf
  T}:\L\ang{\tau}_{\delta} \to \L\ang{\tau}_{\delta}$ by $\text{\sf
  T}(U) := \tau U \bmod F$.  We are interested in computing $\text{\sf
  T}^r(U)$, for some $r \ge 0$ and $U$ in
$\L\ang{\tau}_{\delta}$.

The operator $\text{\sf T}$ is $\F_q$-linear but not $\L$-linear;
the coefficient vector of $\text{\sf T}(U)$ is ${\bm M}\, \pi({\bm
  v}_U)$, where ${\bm M}$ is the companion matrix of $F$ (seen as a
commutative polynomial), ${\bm v}_U$ is the coefficient vector of $U$
and where we still denote by $\pi$ the entry-wise application of the
Frobenius to a vector (or to a matrix). As a result, the coefficient
vector of $\text{\sf T}^r(U)$ is ${\bm M}\, \pi({\bm M}) \cdots
\pi^{r-1}({\bm M})\, \pi^r({\bm v}_U)$.

Lemma~5.3 in~\cite{vonzurGathen1992} shows how to compute such an
expression in $O(\log r)$ applications of Frobenius powers (to
matrices) and matrix products (the original reference deals with
scalars, but there is no difference in the matrix case). When $r$ is
$O(n)$, the runtime is $(\delta^3 n^\theta \log q)^{1+o(1)}$ bit
operations ($\delta$ will be small later on, so there is no need
to use fast matrix arithmetic).

\smallskip\noindent{\bf 6.3.} We will also have to invert $\text{\sf
  T}$.  In order to be able do so, we assume that the constant
coefficient of $F$ is non-zero; as a result, ${\bm M}$ is invertible.
Given $V=\text{\sf T}(U)$, we can recover the coefficient vector
of $U=\text{\sf T}^{-1}(V)$ as ${\bm N}\, \pi^{-1}({\bm v}_V)$, where
${\bm N}=\pi^{-1}({\bm M}^{-1})$. For $r$ in $O(n)$, we can compute
$\text{\sf T}^{-r}(V)$ in $(\delta^3 n^\theta \log q)^{1+o(1)}$ bit
operations as well, replacing the applications of powers of $\pi$ by
powers of $\pi^{-1}$.

\smallskip\noindent{\bf 6.4.} Using the results in {\bf 6.2} and {\bf
  6.3}, let us show how to compute $A \bmod E$, for some irreducible
$E$ in $\F_q[x]$. As input, assume that we know $E$ and $B \bmod
E$. We let $F=\varphi_E \in \L\ang{\tau}$ and $\delta := \deg F = 2
\deg E$. We suppose that $E(\gamma(x)) \ne 0$; as a result, the
constant coefficient of $F$ is non-zero, so {\bf 6.3} applies.

Start from the characteristic equation $\tau^{2n} - \tau^n \varphi_A +
\varphi_B=0$, which we rewrite as $\tau^n \varphi_A = \tau^{2n} +
\varphi_B$ and reduce both sides modulo $F=\varphi_E$. On the left, we
obtain $(\tau^n \varphi_A) \bmod F = \text{\sf T}^n(\varphi_A \bmod
F)$, that is, $\text{\sf T}^n(\varphi_{A \bmod E})$. Similarly, on the
right, we obtain $\text{\sf T}^{2n}(1) + \varphi_{B \bmod E}.$ Thus, we can
proceed as follows:
\begin{enumerate}
\item Compute $F:=\varphi_E$ and $V_0:=\varphi_{B \bmod E}$. By
  {\bf 3.2.3}, the cost is
  $(\delta^{(\omega+1)/2} n^\theta \log q)^{1+o(1)}$ bit operations.
\item Compute the companion matrix ${\bm M}$ of $F$ in 
  $(\delta^2 n \log q)^{1+o(1)}$ bit operations.
\item Compute $V_1:=\text{\sf T}^{2n}(1)$ in $(\delta^3 n^\theta \log
  q)^{1+o(1)}$ bit operations ({\bf 6.2}).
\item Compute $\varphi_{A \bmod E}=\text{\sf T}^{-n}(V_0+V_1)$ in $(\delta^3
  n^\theta \log q)^{1+o(1)}$ bit operations ({\bf 6.3}).
\item Deduce $A \bmod E$ in $(\delta^2 n^\theta \log q)^{1+o(1)}$ bit
  operations ({\bf 3.2.5}).
\end{enumerate}
The overall runtime is $(\delta^3 n^\theta \log q)^{1+o(1)}$ bit
operations.

\smallskip\noindent{\bf 6.5.} We can finally present the
whole algorithm. 
%% \todo{Consider splitting case where $q  \geq \frac{n}{2}$}
\begin{enumerate}
\item Compute the Frobenius norm $B$ (Proposition~\ref{frobnorm})
\item Compute polynomials $E_1,\dots,E_s$ as in {\bf 6.1}.
\item For $i=1,\dots,s$, compute $B_i:=B \bmod E_i$.
\item For $i=1,\dots,s$, compute $A_i:=A \bmod E_i$ by {\bf 6.4}.
\item Recover $A$ by the Chinese remainder map.
\end{enumerate}
Steps (1), (2), (3) and (5) take a total of $(n^2 \log q)^{1+o(1)}$
bit operations. Since the degrees of all polynomials $E_i$ are $O(\log
n)$, the time spent at Step (4) is $(n^{\theta+1} \log
q)^{1+o(1)}$ bit operations. Since we can take $\theta = 1+\varepsilon$ 
for any $\varepsilon > 0$, and adding the cost
$(n \log^2 q)^{1+o(1)}$
 of computing $x^q \bmod \frakp$, 
this establishes the second statement in Theorem~\ref{theo:main}.
%% \todo{discard $E_i$ if $E_i(\gamma(x)) =0$}

%%%%%%%%%%%%%%%%%%%%%%%%%%%%%%%%%%%%%%%%%%%%%%%%%%%%%%%%%%%%
%%%%%%%%%%%%%%%%%%%%%%%%%%%%%%%%%%%%%%%%%%%%%%%%%%%%%%%%%%%%
%%%%%%%%%%%%%%%%%%%%%%%%%%%%%%%%%%%%%%%%%%%%%%%%%%%%%%%%%%%%

\section{A Monte Carlo Algorithm}\label{sec:mc}

We now prove the last item in our main theorem: {\em
  there exists a Monte Carlo algorithm that solves Problem~\ref{pb1}
  in $(n^2 \log^2 q)^{1+o(1)}$ bit operations.} The runtime is now
quadratic in $\log q$, but truly quadratic in $n$, not of the form
$n^{2+\varepsilon}$. The point is that we avoid applying high
powers of the Frobenius (and thus modular composition); the
applications of $\Phi_x= \gamma(x){\rm Id} + g \pi + \Delta \pi^2$ are
done by repeated squaring.  This algorithm behaves well in practice,
whereas the behavior of modular composition significantly
hinders the implementation of the algorithms in the previous sections;
see {\bf 3.1.3} and {\bf 3.2.4}.

The algorithm is inspired by~\cite[Th.~5]{Shoup94}; it bears
similarities with Narayanan's, but does not require the assumption
that the minimal polynomial $\Gamma$ of $\Phi_x= \gamma(x){\rm Id} + g \pi + \Delta \pi^2$ have degree
$n$. Whether the subquadratic runtime obtained in
Section~\ref{sec:narayanan} can be carried over to the approach
presented here is of course an interesting question.

\smallskip\noindent{\bf 7.1.} When $n$ is even, we may need to
determine the leading coefficient $a_{n/2}$ of the Frobenius trace $A$
separately. We will use the following result, due to
Jung~\cite{Jung00,frobdist}:
\[a_{n/2} = {\rm Tr}_{\F_{q^2}/\F_{q}}({\rm N}_{\L/\F_{q^2}}(\Delta)^{-1}),\]
where $\F_{q^2}$ is the unique degree 2 extension of $\F_q$ contained
in $\L$, and ${\rm Tr}$ and ${\rm N}$ are (finite field) trace and
norm. Using repeated squaring for exponentation, $a_{n/2}$ can be
computed in $(n^2 \log q)^{1+o(1)}$ operations in $\F_q$, so $(n^2
  \log^2 q)^{1+o(1)}$ bit operations.
%(Other solutions are possible, using for instance a resultant
 %for norm computation as in Section~\ref{ssec:gek}.)

\smallskip\noindent{\bf 7.2.} Let $\Gamma \in \F_q[x]$ be the minimal
polynomial of $\Phi_x$ and let $\nu \le n$ its degree. We prove here 
that the inequality $\nu\ge n/2$ holds.

For any positive integers $i,j$ with $0 \le i < j < n$, $\pi^i \ne
\pi^j$. Therefore, by independence of characters, ${\rm Id}, \pi,
\ldots, \pi^{n-1}$ satisfy no non-trivial $\L$-linear relation;
that is, there are no constants $c_0, \ldots, c_{n-1}$ in $\L$,
with at least one $c_i \neq 0$, such that $c_0 + c_1 \pi + \ldots +
c_{n-1}\pi^{n-1}=0$ in ${\rm End}_{\F_q}[\L]$.

Assume by way of contradiction that $2\nu \le n-1$.  We know that
$\Gamma(\Phi_x) = 0$; since $\Gamma$ has degree $\nu$, we may write
is as $\Gamma = c_0 + \cdots + c_{\nu-1} x^{\nu-1} +
x^\nu$. Evaluating at $\Phi_x = \gamma(x) {\rm Id} + g \pi + \Delta
\pi^2$, we obtain a relation of the form $\bar c_0 {\rm Id} + \bar
c_1 \pi + \cdots + \bar c_{2\nu} \pi^{2\nu} = 0$ with
coefficients in $\L$, where all exponents are at most $n-1$. The
leading coefficient $\bar c_{2\nu}$ is given by $\bar c_{2\nu} =
\Delta^{(1-q^{2\nu})/(1-q)}$, so it is non-zero, a contradiction.
Thus, $2\nu \ge n$, as claimed.

\smallskip\noindent{\bf 7.3.} The first step in the algorithm computes
the minimal polynomial $\Gamma$ of $\Phi_x$. 
To do so, choose at random $\alpha$ in $\L$ and an $\F_q$-linear
projection map  $\ell: \L \to \F_q$. The sequence
$(\ell(\Phi_x^i(\alpha)))_{i \ge 0}$ is linearly generated, and its
minimal polynomial $\Gamma_{\ell,\alpha}$ divides $\Gamma$. Given $2n$
entries in the sequence $\ell(\Phi_x^i(\alpha))$, we apply the
Berlekamp-Massey algorithm to obtain $\Gamma_{\ell,\alpha}$.

Assuming that $\ell$ and $\alpha$ are chosen uniformly at random,
Wiedemann proved~\cite{Wiedemann:1986:SSL:13738.13744} that the
probability that $\Gamma_{\ell,\alpha}=\Gamma$ is at least $1/(12
\max(1, \log_q \nu))$. Using the DeMillo-Lipton-Zippel-Schwartz lemma
gives another lower bound for the probability that
$\Gamma_{\ell,\alpha}$ equals $\Gamma$, namely
$1-2n/q$~\cite{Kaltofen:1991:PEP:113379.113396,Kaltofen-saun:1991:WMS:646027.676885}. We
will assume henceforth that this is the case
(as in Section~\ref{sec:narayanan}, we can work over an extension 
field of $\F_q$ of degree $O(\log n)$ if $q < n$).

\smallskip\noindent{\bf 7.4.}
We start from $A = \sum_{i = 0}^{\floor*{n/2}}a_ix^i \in \F_q[x]$, for
some unknown coefficients $a_i$. Since $n/2 \le \nu$ (by {\bf 7.2}), we must have 
$\floor*{n/2} \le \nu -1$, except if $n$ is even and $n/2 = \nu$.
Hence, we may rewrite $A$ as 
$$A = \sum_{i = 0}^{\nu - 1}a_ix^i + a_{\nu}x^{\nu},$$ where $a_i = 0$
for $i=\floor*{n/2}+1,\dots,\nu-1$ and either $a_\nu=0$ (if
$\floor*{n/2} \le \nu -1$) or $a_\nu$ can be determined as in {\bf
  7.1} (if $\floor*{n/2} = \nu$). In any case, $a_\nu$ is known.

Theorem~\ref{charpoly} implies that for $\alpha$ as above, we have $\Phi_A(\alpha)=r$ with $ r:=\alpha + \Phi_B(\alpha) \in \L$. Using the expression of $A$ given above, this yields
\[ \sum_{i = 0}^{\nu - 1}a_i \Phi_{x^{i}}(\alpha) = \tilde r,\]
with $\tilde r =  r-a_\nu \Phi_\nu(\alpha)$.
For $j \ge 0$, applying $\Phi_{x^j}$ to this equality
gives
\[ \sum_{i = 0}^{\nu - 1}a_i \Phi_{x^{i+j}}(\alpha) = \Phi_{x^j}(\tilde r).\]
Finally, we can apply $\ell$ to both sides of such equalities,
for $j=0,\dots,\nu-1$.
This yields the following Hankel system:
\begin{align}\label{eq:A}
  \begin{bmatrix}   \ell(\alpha) & \ldots & \ell(\Phi_{x^{\nu-1}}(\alpha)) \\
    \vdots & & \vdots  \\ 
  \ell(\Phi_{x^{\nu-1}}(\alpha)) &  \ldots & \ell(\Phi_{x^{2\nu-2}}(\alpha))
\end{bmatrix} 
\begin{bmatrix} a_0  \\ \vdots \\ a_{\nu-1} \end{bmatrix} 
= 
\begin{bmatrix} \ell(\tilde r) \\ \vdots \\   \ell(\Phi_{x^{\nu-1}}(\tilde r)) \end{bmatrix}. 
\end{align}
Since we assumed that $\Gamma_{\ell,\alpha}=\Gamma$, applying for
instance Lemma~1 in~\cite{Kaltofen:1991:PEP:113379.113396}, we deduce 
that the matrix of the system is invertible, allowing us to
recover $a_0,\dots,a_{\nu-1}$.

\smallskip\noindent{\bf 7.5.} We can now summarize the algorithm and
analyze its runtime.
\begin{enumerate}
\item Compute the Frobenius norm $B=\sum_{i \le n} b_i x^i$ (Proposition~\ref{frobnorm});
this takes $(n \log q)^{1+o(1)}$ bit operations.
\item Compute the sequence $(\Phi_{x^{i}}(\alpha))_{i < 2n}$ using the
  recurrence relation $\Phi_{x^{i+1}}(\alpha) = (\gamma(x) {\rm Id} +
  g\pi + \Delta \pi^2)(\Phi_{x^i}(\alpha))$. Using repeated squaring to
  apply the Frobenius, we get all terms in $(n^2 \log^2 q)^{1+o(1)}$ bit
  operations.
\item Apply $\ell$ to all terms of the sequence and deduce
  $\Gamma_{\ell,\alpha}$ by the Berlekamp-Massey algorithm. This takes
  $(n^2 \log q)^{1+o(1)}$ bit operations. We assume
  $\Gamma_{\ell,\alpha}=\Gamma$ and let $\nu$ be its degree.
\item If $n$ is even and $\nu=n/2$, compute $a_\nu$ as in {\bf 7.1};
  otherwise, set $a_\nu=0$. This takes $(n^2 \log^2 q)^{1+o(1)}$ bit
  operations.
\item Compute $\tilde r = \alpha + \sum_{i \le n} b_i
  \Phi_{x^{i}}(\alpha) - a_\nu \Phi_\nu(\alpha)$; this takes $(n^2
  \log^2 q)^{1+o(1)}$ bit operations.
\item Compute the sequence $(\Phi_{x^{i}}(\tilde r))_{i < \nu}$
  and apply $\ell$ to all entries in this sequence. As above, this takes
 $(n^2 \log^2 q)^{1+o(1)}$ bit
  operations.
\item Solve~\eqref{eq:A}; since the matrix is Hankel and non-singular,
  this takes $(n^2 \log q)^{1+o(1)}$ bit operations.
\end{enumerate}
Altogether, this takes  $(n^2 \log^2 q)^{1+o(1)}$ bit operations,
as claimed.
%% \todo{table comparing the different approaches?}

%%%%%%%%%%%%%%%%%%%%%%%%%%%%%%%%%%%%%%%%%%%%%%%%%%%%%%%%%%%%%%%%%%%%%%%
%%%%%%%%%%%%%%%%%%%%%%%%%%%%%%%%%%%%%%%%%%%%%%%%%%%%%%%%%%%%%%%%%%%%%%%
%%%%%%%%%%%%%%%%%%%%%%%%%%%%%%%%%%%%%%%%%%%%%%%%%%%%%%%%%%%%%%%%%%%%%%%

\section{Experimental Results}\label{sec:experiments}

In support of our theoretical analysis, the algorithms presented in
sections 6 and 7, as well as Gekeler's algorithm
in~\cite[Section~3]{frobdist}, were implemented in C++ using Shoup's
NTL library~\cite{shoup2001ntl}; our implementation currently supports
 prime $q$.  When $m=1$, we also compare our implementation
with that of the algorithm in~\cite{eschost2017arXiv171200669D}.

Table \ref{tab:table-sample} provides sample runtimes for several
parameters. Figure 1 is made up of 24 data points for $q = 499$, $m =
2$, and varied $n$, averaged over 4 runs. The randomized algorithm of
section 7 demonstrated a significant runtime advantage over both
Gekeler's original algorithm and the deterministic alternative. Due to
its heavy dependency on modular composition, and a lack of readily
available implementations of the Kedlaya-Umans algorithm on which we
rely, the deterministic algorithm demonstrates a
significantly higher complexity than expected. For $m=1$, as predicted by the cost 
analysis, the algorithm in~\cite{eschost2017arXiv171200669D} is overall
the fastest.

The code used to generate these results is publicly available at \url{https://github.com/ymusleh/Drinfeld-paper/tree/master/code}.

\begin{center}
\begin{tabular}{ | m{1.8cm} | m{1.5cm}| m{1.6cm} | m{0.9cm} | m{1cm} | } 
%\caption{A summary of the computational results}
\hline
 & Randomized & Deterministic & Gekeler & Hasse~\cite{eschost2017arXiv171200669D}  \\ 
\hline
$q = 571$, $n=32$, $m = 1$ & 0.065341 &	1.19282 &	0.37291 & 0.02087\\
\hline 
$q = 571$, $n=32$, $m = 4$  & 0.060729 & 1.17222 & 0.38341 & - \\ 
\hline
$q = 850853$, $n = 64$, $m=1$ & 0.598883 & 25.4512 & 8.97046 & 0.12814  \\
\hline 
$q = 850853$, $n = 64$, $m=8$ &0.615858 & 25.6425 & 9.12075 & - \\
\hline
\end{tabular}
\end{center}
\begin{table}[h!]
  \centering
  \caption{Various parameter test cases; time in seconds. }
  \label{tab:table-sample}
\end{table}

\vspace{-1cm}

\begin{figure}[h!]\label{fig:ntest499}
\centering
  \includegraphics[width=3in]{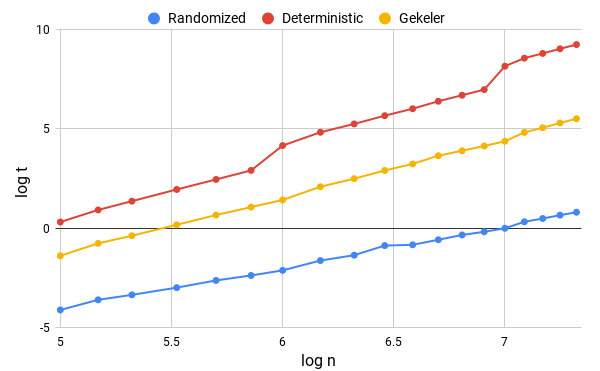}
  \caption{Log-log plot of $n$ versus runtime with $q = 499$, $m = 2$}
\end{figure}

%\vspace{-0.5cm}

%\begin{figure}[h!]\label{fig:ptest50}
%\centering
%  \includegraphics[width=3in]{chart-50-2.png}
%  \caption{Log-log plot of $q$ versus runtime with %$n = 50$, $m = 2$}
%\end{figure}

\begin{acks}
  We wish to thank Jason Bell and Mark Giesbrecht for their comments
  on Musleh's MMath thesis \cite{Musleh}, which is the basis of this work, and Anand
Kumar  Narayanan for answering many of our questions. Schost was supported
  by an NSERC Discovery Grant.
\end{acks}

%%%%%%%%%%%%%%%%%%%%%%%%%%%%%%%%%%%%%%%%%%%%%%%%%%%%%%%%%%%%%%%%%%%%%%%%

\bibliographystyle{ACM-Reference-Format}
\bibliography{drinfeld_charpoly}


\begin{thebibliography}{39}

%%% the Web address.

\ifx \showCODEN    \undefined \def \showCODEN     #1{\unskip}     \fi
\ifx \showDOI      \undefined \def \showDOI       #1{#1}\fi
\ifx \showISBNx    \undefined \def \showISBNx     #1{\unskip}     \fi
\ifx \showISBNxiii \undefined \def \showISBNxiii  #1{\unskip}     \fi
\ifx \showISSN     \undefined \def \showISSN      #1{\unskip}     \fi
\ifx \showLCCN     \undefined \def \showLCCN      #1{\unskip}     \fi
\ifx \shownote     \undefined \def \shownote      #1{#1}          \fi
\ifx \showarticletitle \undefined \def \showarticletitle #1{#1}   \fi
\ifx \showURL      \undefined \def \showURL       {\relax}        \fi
% invisible to TeX
\providecommand\bibfield[2]{#2}
\providecommand\bibinfo[2]{#2}
\providecommand\natexlab[1]{#1}
\providecommand\showeprint[2][]{arXiv:#2}

\bibitem[\protect\citeauthoryear{Bostan, Gaudry, and Schost}{Bostan
  et~al\mbox{.}}{2007}]%
        {BoGaSc07}
\bibfield{author}{\bibinfo{person}{A. Bostan}, \bibinfo{person}{P. Gaudry},
  {and} \bibinfo{person}{{\'E}. Schost}.} \bibinfo{year}{2007}\natexlab{}.
\newblock \showarticletitle{Linear recurrences with polynomial coefficients and
  application to integer factorization and {C}artier-{M}anin operator}.
\newblock \bibinfo{journal}{\emph{SIAM J. Comput.}} \bibinfo{volume}{36},
  \bibinfo{number}{6} (\bibinfo{year}{2007}), \bibinfo{pages}{1777--1806}.
\newblock


\bibitem[\protect\citeauthoryear{Brent and Kung}{Brent and Kung}{1978}]%
        {BrKu78}
\bibfield{author}{\bibinfo{person}{R.~P. Brent} {and} \bibinfo{person}{H.~T.
  Kung}.} \bibinfo{year}{1978}\natexlab{}.
\newblock \showarticletitle{Fast Algorithms for Manipulating Formal Power
  Series}.
\newblock \bibinfo{journal}{\emph{J. ACM}} \bibinfo{volume}{25},
  \bibinfo{number}{4} (\bibinfo{year}{1978}), \bibinfo{pages}{581--595}.
\newblock
\showISSN{0004-5411}


\bibitem[\protect\citeauthoryear{Carlitz}{Carlitz}{1935}]%
        {Carlitz35}
\bibfield{author}{\bibinfo{person}{L. Carlitz}.}
  \bibinfo{year}{1935}\natexlab{}.
\newblock \showarticletitle{On certain functions connected with polynomials in
  a Galois field}.
\newblock \bibinfo{journal}{\emph{Duke Math. J.}} \bibinfo{volume}{1},
  \bibinfo{number}{2} (\bibinfo{year}{1935}), \bibinfo{pages}{137--168}.
\newblock


\bibitem[\protect\citeauthoryear{Caruso and Borgne}{Caruso and Borgne}{2017}]%
        {CaLe17}
\bibfield{author}{\bibinfo{person}{X. Caruso} {and} \bibinfo{person}{J.~Le
  Borgne}.} \bibinfo{year}{2017}\natexlab{}.
\newblock \showarticletitle{Fast multiplication for skew polynomials}. In
  \bibinfo{booktitle}{\emph{ISSAC'17}}. \bibinfo{publisher}{ACM},
  \bibinfo{pages}{77--84}.
\newblock


\bibitem[\protect\citeauthoryear{Coppersmith}{Coppersmith}{1994}]%
        {Coppersmith94}
\bibfield{author}{\bibinfo{person}{D. Coppersmith}.}
  \bibinfo{year}{1994}\natexlab{}.
\newblock \showarticletitle{Solving homogeneous linear equations over {GF}$(2)$
  via block {W}iedemann algorithm}.
\newblock \bibinfo{journal}{\emph{Math. Comp.}} \bibinfo{volume}{62},
  \bibinfo{number}{205} (\bibinfo{year}{1994}), \bibinfo{pages}{333--350}.
\newblock


\bibitem[\protect\citeauthoryear{Coppersmith and Winograd}{Coppersmith and
  Winograd}{1990}]%
        {CoWi90}
\bibfield{author}{\bibinfo{person}{D. Coppersmith} {and} \bibinfo{person}{S.
  Winograd}.} \bibinfo{year}{1990}\natexlab{}.
\newblock \showarticletitle{Matrix multiplication via arithmetic progressions}.
\newblock \bibinfo{journal}{\emph{J. Symb. Comput.}} \bibinfo{volume}{9},
  \bibinfo{number}{3} (\bibinfo{year}{1990}), \bibinfo{pages}{251--280}.
\newblock


\bibitem[\protect\citeauthoryear{{Doliskani}, {Narayanan}, and
  {Schost}}{{Doliskani} et~al\mbox{.}}{2017}]%
        {eschost2017arXiv171200669D}
\bibfield{author}{\bibinfo{person}{J. {Doliskani}}, \bibinfo{person}{A.~K.
  {Narayanan}}, {and} \bibinfo{person}{{\'E}. {Schost}}.}
  \bibinfo{year}{2017}\natexlab{}.
\newblock \bibinfo{title}{{Drinfeld modules with complex multiplication,
  {H}asse invariants and factoring polynomials over finite fields}}.
\newblock
\newblock
\showeprint{1712.00669}


\bibitem[\protect\citeauthoryear{Drinfel'd}{Drinfel'd}{1974}]%
        {Drinfeld74}
\bibfield{author}{\bibinfo{person}{V.~G. Drinfel'd}.}
  \bibinfo{year}{1974}\natexlab{}.
\newblock \showarticletitle{Elliptic modules}.
\newblock \bibinfo{journal}{\emph{Matematicheskii Sbornik}}
  \bibinfo{volume}{94}, \bibinfo{number}{23} (\bibinfo{year}{1974}),
  \bibinfo{pages}{561--593}.
\newblock


\bibitem[\protect\citeauthoryear{Garai and Papikian}{Garai and
  Papikian}{2018}]%
        {GaPa18}
\bibfield{author}{\bibinfo{person}{S. Garai} {and} \bibinfo{person}{M.
  Papikian}.} \bibinfo{year}{2018}\natexlab{}.
\newblock \bibinfo{title}{Endomorphism rings of reductions of Drinfeld
  modules}.
\newblock
\newblock
\showeprint{1804.07904}


\bibitem[\protect\citeauthoryear{Gathen and Gerhard}{Gathen and
  Gerhard}{2013}]%
        {Gathen:2003:MCA:945759}
\bibfield{author}{\bibinfo{person}{J.~von~zur Gathen} {and} \bibinfo{person}{J.
  Gerhard}.} \bibinfo{year}{2013}\natexlab{}.
\newblock \bibinfo{booktitle}{\emph{Modern Computer Algebra}
  (\bibinfo{edition}{3} ed.)}.
\newblock \bibinfo{publisher}{Cambridge University Press},
  \bibinfo{address}{New York, NY, USA}.
\newblock
\showISBNx{0521826462}


\bibitem[\protect\citeauthoryear{Gathen and Shoup}{Gathen and Shoup}{1992}]%
        {vonzurGathen1992}
\bibfield{author}{\bibinfo{person}{J.~von~zur Gathen} {and} \bibinfo{person}{V.
  Shoup}.} \bibinfo{year}{1992}\natexlab{}.
\newblock \showarticletitle{Computing Frobenius maps and factoring
  polynomials}.
\newblock \bibinfo{journal}{\emph{Computational Complexity}}
  \bibinfo{volume}{2}, \bibinfo{number}{3} (\bibinfo{year}{1992}),
  \bibinfo{pages}{187--224}.
\newblock


\bibitem[\protect\citeauthoryear{Gekeler}{Gekeler}{1991}]%
        {GEKELE1991187}
\bibfield{author}{\bibinfo{person}{E.-U. Gekeler}.}
  \bibinfo{year}{1991}\natexlab{}.
\newblock \showarticletitle{On finite Drinfeld modules}.
\newblock \bibinfo{journal}{\emph{Journal of Algebra}} \bibinfo{volume}{141},
  \bibinfo{number}{1} (\bibinfo{year}{1991}), \bibinfo{pages}{187 -- 203}.
\newblock
\showISSN{0021-8693}


\bibitem[\protect\citeauthoryear{Gekeler}{Gekeler}{2008}]%
        {frobdist}
\bibfield{author}{\bibinfo{person}{E.-U. Gekeler}.}
  \bibinfo{year}{2008}\natexlab{}.
\newblock \showarticletitle{Frobenius distributions of Drinfeld modules over
  finite fields}.
\newblock \bibinfo{journal}{\emph{Trans. Amer. Math. Soc.}}
  \bibinfo{volume}{360} (\bibinfo{date}{04} \bibinfo{year}{2008}),
  \bibinfo{pages}{1695--1721}.
\newblock


\bibitem[\protect\citeauthoryear{Giorgi, Jeannerod, and Villard}{Giorgi
  et~al\mbox{.}}{2003}]%
        {GiJeVi03}
\bibfield{author}{\bibinfo{person}{P. Giorgi}, \bibinfo{person}{C.-P.
  Jeannerod}, {and} \bibinfo{person}{G. Villard}.}
  \bibinfo{year}{2003}\natexlab{}.
\newblock \showarticletitle{On the complexity of polynomial matrix
  computations}. In \bibinfo{booktitle}{\emph{ISSAC'03}}.
  \bibinfo{publisher}{ACM}, \bibinfo{pages}{135--142}.
\newblock
\showISBNx{1-58113-641-2}


\bibitem[\protect\citeauthoryear{Goss}{Goss}{1996}]%
        {Goss96}
\bibfield{author}{\bibinfo{person}{D. Goss}.} \bibinfo{year}{1996}\natexlab{}.
\newblock \bibinfo{booktitle}{\emph{Basic Structures of Function Field
  Arithmetic}}.
\newblock \bibinfo{publisher}{Springer Berlin Heidelberg}.
\newblock


\bibitem[\protect\citeauthoryear{Harvey}{Harvey}{2014}]%
        {Harvey14}
\bibfield{author}{\bibinfo{person}{David Harvey}.}
  \bibinfo{year}{2014}\natexlab{}.
\newblock \showarticletitle{Counting points on hyperelliptic curves in average
  polynomial time}.
\newblock \bibinfo{journal}{\emph{Annals of Mathematics}}
  \bibinfo{volume}{179}, \bibinfo{number}{2} (\bibinfo{year}{2014}),
  \bibinfo{pages}{783--803}.
\newblock


\bibitem[\protect\citeauthoryear{Hsia and Yu}{Hsia and Yu}{2000}]%
        {HsYu00}
\bibfield{author}{\bibinfo{person}{L.-C. Hsia} {and} \bibinfo{person}{J. Yu}.}
  \bibinfo{year}{2000}\natexlab{}.
\newblock \showarticletitle{On characteristic polynomials of geometric
  {F}robenius associated to {D}rinfeld modules}.
\newblock \bibinfo{journal}{\emph{Compositio Mathematica}}
  \bibinfo{volume}{122}, \bibinfo{number}{3} (\bibinfo{year}{2000}),
  \bibinfo{pages}{261--280}.
\newblock


\bibitem[\protect\citeauthoryear{Jung}{Jung}{2000}]%
        {Jung00}
\bibfield{author}{\bibinfo{person}{F. Jung}.} \bibinfo{year}{2000}\natexlab{}.
\newblock \bibinfo{title}{Charakteristische {P}olynome von
  {D}rinfeld-{M}oduln}.
\newblock
\newblock
\newblock
\shownote{Diplomarbeit, U. Saarbr\"ucken.}


\bibitem[\protect\citeauthoryear{Kaltofen and Pan}{Kaltofen and Pan}{1991}]%
        {Kaltofen:1991:PEP:113379.113396}
\bibfield{author}{\bibinfo{person}{E. Kaltofen} {and} \bibinfo{person}{V.
  Pan}.} \bibinfo{year}{1991}\natexlab{}.
\newblock \showarticletitle{Processor efficient parallel solution of linear
  systems over an abstract field}. In \bibinfo{booktitle}{\emph{SPAA '91}}.
  \bibinfo{publisher}{ACM}, \bibinfo{pages}{180--191}.
\newblock
\showISBNx{0-89791-438-4}


\bibitem[\protect\citeauthoryear{Kaltofen and Saunders}{Kaltofen and
  Saunders}{1991}]%
        {Kaltofen-saun:1991:WMS:646027.676885}
\bibfield{author}{\bibinfo{person}{E. Kaltofen} {and} \bibinfo{person}{B.~D.
  Saunders}.} \bibinfo{year}{1991}\natexlab{}.
\newblock \showarticletitle{On Wiedemann's method of solving sparse linear
  systems}. In \bibinfo{booktitle}{\emph{AAECC-9}}.
  \bibinfo{publisher}{Springer-Verlag}, \bibinfo{pages}{29--38}.
\newblock
\showISBNx{3-540-54522-0}


\bibitem[\protect\citeauthoryear{Kaltofen and Shoup}{Kaltofen and
  Shoup}{1998}]%
        {KaSh98}
\bibfield{author}{\bibinfo{person}{E. Kaltofen} {and} \bibinfo{person}{V.
  Shoup}.} \bibinfo{year}{1998}\natexlab{}.
\newblock \showarticletitle{Subquadratic-time factoring of polynomials over
  finite fields}.
\newblock \bibinfo{journal}{\emph{Math. Comp.}} \bibinfo{volume}{67},
  \bibinfo{number}{223} (\bibinfo{year}{1998}), \bibinfo{pages}{1179--1197}.
\newblock


\bibitem[\protect\citeauthoryear{Kaltofen and Villard}{Kaltofen and
  Villard}{2004}]%
        {KaVi04}
\bibfield{author}{\bibinfo{person}{E. Kaltofen} {and} \bibinfo{person}{G.
  Villard}.} \bibinfo{year}{2004}\natexlab{}.
\newblock \showarticletitle{On the complexity of computing determinants}.
\newblock \bibinfo{journal}{\emph{Computational Complexity}}
  \bibinfo{volume}{13}, \bibinfo{number}{3-4} (\bibinfo{year}{2004}),
  \bibinfo{pages}{91--130}.
\newblock


\bibitem[\protect\citeauthoryear{Kaminski, Kirkpatrick, and Bshouty}{Kaminski
  et~al\mbox{.}}{1988}]%
        {KaKiBs88}
\bibfield{author}{\bibinfo{person}{M. Kaminski}, \bibinfo{person}{D.G.
  Kirkpatrick}, {and} \bibinfo{person}{N.H. Bshouty}.}
  \bibinfo{year}{1988}\natexlab{}.
\newblock \showarticletitle{Addition requirements for matrix and transposed
  matrix products}.
\newblock \bibinfo{journal}{\emph{J. Algorithms}} \bibinfo{volume}{9},
  \bibinfo{number}{3} (\bibinfo{year}{1988}), \bibinfo{pages}{354--364}.
\newblock


\bibitem[\protect\citeauthoryear{Kedlaya and Umans}{Kedlaya and Umans}{2011}]%
        {Kedlaya:2011:FPF:2340436.2340448}
\bibfield{author}{\bibinfo{person}{K.~S. Kedlaya} {and} \bibinfo{person}{C.
  Umans}.} \bibinfo{year}{2011}\natexlab{}.
\newblock \showarticletitle{Fast polynomial factorization and modular
  composition}.
\newblock \bibinfo{journal}{\emph{SIAM J. Comput.}} \bibinfo{volume}{40},
  \bibinfo{number}{6} (\bibinfo{year}{2011}), \bibinfo{pages}{1767--1802}.
\newblock


\bibitem[\protect\citeauthoryear{Labahn, Neiger, and Zhou}{Labahn
  et~al\mbox{.}}{2017}]%
        {LaNeZh17}
\bibfield{author}{\bibinfo{person}{G. Labahn}, \bibinfo{person}{V. Neiger},
  {and} \bibinfo{person}{W. Zhou}.} \bibinfo{year}{2017}\natexlab{}.
\newblock \showarticletitle{Fast, deterministic computation of the {H}ermite
  normal form and determinant of a polynomial matrix}.
\newblock \bibinfo{journal}{\emph{J. Complexity}}  \bibinfo{volume}{42}
  (\bibinfo{year}{2017}), \bibinfo{pages}{44--71}.
\newblock


\bibitem[\protect\citeauthoryear{Le~Gall}{Le~Gall}{2014}]%
        {LeGall14}
\bibfield{author}{\bibinfo{person}{F. Le~Gall}.}
  \bibinfo{year}{2014}\natexlab{}.
\newblock \showarticletitle{Powers of tensors and fast matrix multiplication}.
  In \bibinfo{booktitle}{\emph{ISSAC'14}}. \bibinfo{publisher}{ACM},
  \bibinfo{pages}{296--303}.
\newblock


\bibitem[\protect\citeauthoryear{{Le Gall} and Urrutia}{{Le Gall} and
  Urrutia}{2018}]%
        {LeUr18}
\bibfield{author}{\bibinfo{person}{F. {Le Gall}} {and} \bibinfo{person}{F.
  Urrutia}.} \bibinfo{year}{2018}\natexlab{}.
\newblock \showarticletitle{Improved rectangular matrix multiplication using
  powers of the {C}oppersmith-{W}inograd tensor}. In
  \bibinfo{booktitle}{\emph{SODA '18}}. \bibinfo{publisher}{SIAM},
  \bibinfo{pages}{1029--1046}.
\newblock


\bibitem[\protect\citeauthoryear{{Musleh, Yossef}}{{Musleh, Yossef}}{2018}]%
        {Musleh}
\bibfield{author}{\bibinfo{person}{{Musleh, Yossef}}.}
  \bibinfo{year}{2018}\natexlab{}.
\newblock \bibinfo{title}{Fast Algorithms for Finding the Characteristic
  Polynomial of a Rank-2 Drinfeld Module}.
\newblock
\newblock
\urldef\tempurl%
\url{http://hdl.handle.net/10012/13889}
\showURL{%
\tempurl}


\bibitem[\protect\citeauthoryear{Narayanan}{Narayanan}{2018}]%
        {Narayanan18}
\bibfield{author}{\bibinfo{person}{A.~K. Narayanan}.}
  \bibinfo{year}{2018}\natexlab{}.
\newblock \showarticletitle{Polynomial factorization over finite fields by
  computing {E}uler-{P}oincar{\'{e}} characteristics of {D}rinfeld modules}.
\newblock \bibinfo{journal}{\emph{Finite Fields Appl.}}  \bibinfo{volume}{54}
  (\bibinfo{year}{2018}), \bibinfo{pages}{335--365}.
\newblock


\bibitem[\protect\citeauthoryear{Panchishkin and Potemine}{Panchishkin and
  Potemine}{1989}]%
        {PaPo89}
\bibfield{author}{\bibinfo{person}{A. Panchishkin} {and} \bibinfo{person}{I
  Potemine}.} \bibinfo{year}{1989}\natexlab{}.
\newblock \showarticletitle{An algorithm for the factorization of polynomials
  using elliptic modules}. In \bibinfo{booktitle}{\emph{Constructive methods
  and algorithms in number theory}}. \bibinfo{pages}{117}.
\newblock


\bibitem[\protect\citeauthoryear{Pohst and Zassenhaus}{Pohst and
  Zassenhaus}{1989}]%
        {Pohst:1989:AAN:76692}
\bibfield{editor}{\bibinfo{person}{M. Pohst} {and} \bibinfo{person}{H.
  Zassenhaus}} (Eds.). \bibinfo{year}{1989}\natexlab{}.
\newblock \bibinfo{booktitle}{\emph{Algorithmic Algebraic Number Theory}}.
\newblock \bibinfo{publisher}{Cambridge University Press}.
\newblock
\showISBNx{0-521-33060-2}


\bibitem[\protect\citeauthoryear{Puchinger and Wachter-Zeh}{Puchinger and
  Wachter-Zeh}{2017}]%
        {PUCHINGER2017b}
\bibfield{author}{\bibinfo{person}{S. Puchinger} {and} \bibinfo{person}{A.
  Wachter-Zeh}.} \bibinfo{year}{2017}\natexlab{}.
\newblock \showarticletitle{Fast operations on linearized polynomials and their
  applications in coding theory}.
\newblock \bibinfo{journal}{\emph{J. Symb. Comput.}} (\bibinfo{year}{2017}).
\newblock


\bibitem[\protect\citeauthoryear{Satoh}{Satoh}{2000}]%
        {Satoh00}
\bibfield{author}{\bibinfo{person}{T Satoh}.} \bibinfo{year}{2000}\natexlab{}.
\newblock \showarticletitle{The canonical lift of an ordinary elliptic curve
  over a finite field and its point counting}.
\newblock \bibinfo{journal}{\emph{J. Ramanujan Math. Soc.}}
  \bibinfo{volume}{15} (\bibinfo{year}{2000}), \bibinfo{pages}{247--270}.
\newblock


\bibitem[\protect\citeauthoryear{Scanlon}{Scanlon}{2001}]%
        {Scanlon01}
\bibfield{author}{\bibinfo{person}{T. Scanlon}.}
  \bibinfo{year}{2001}\natexlab{}.
\newblock \showarticletitle{Public Key cryptosystems based on Drinfeld modules
  Are insecure}.
\newblock \bibinfo{journal}{\emph{Journal of Cryptology}} \bibinfo{volume}{14},
  \bibinfo{number}{4} (\bibinfo{year}{2001}), \bibinfo{pages}{225--230}.
\newblock


\bibitem[\protect\citeauthoryear{Schoof}{Schoof}{1985}]%
        {schoof85}
\bibfield{author}{\bibinfo{person}{R. Schoof}.}
  \bibinfo{year}{1985}\natexlab{}.
\newblock \showarticletitle{Elliptic curves over finite fields and the
  computation of square roots $\operatorname{mod} p$}.
\newblock \bibinfo{journal}{\emph{Math. Comp.}} \bibinfo{volume}{44},
  \bibinfo{number}{170} (\bibinfo{year}{1985}), \bibinfo{pages}{483--494}.
\newblock
\showISSN{00255718, 10886842}


\bibitem[\protect\citeauthoryear{Shoup}{Shoup}{1994}]%
        {Shoup94}
\bibfield{author}{\bibinfo{person}{V. Shoup}.} \bibinfo{year}{1994}\natexlab{}.
\newblock \showarticletitle{Fast construction of irreducible polynomials over
  finite fields}.
\newblock \bibinfo{journal}{\emph{J. Symb. Comput.}} \bibinfo{volume}{17},
  \bibinfo{number}{5} (\bibinfo{year}{1994}), \bibinfo{pages}{371--391}.
\newblock


\bibitem[\protect\citeauthoryear{Shoup}{Shoup}{2019}]%
        {shoup2001ntl}
\bibfield{author}{\bibinfo{person}{V. Shoup}.} \bibinfo{year}{2019}\natexlab{}.
\newblock \bibinfo{title}{{NTL}: A library for doing number theory}.
\newblock
\newblock
\newblock
\shownote{\url{http:/www.shoup.net/ntl}.}


\bibitem[\protect\citeauthoryear{van~der Heiden}{van~der Heiden}{2004}]%
        {vanderHeiden04}
\bibfield{author}{\bibinfo{person}{G.~J. van~der Heiden}.}
  \bibinfo{year}{2004}\natexlab{}.
\newblock \showarticletitle{Factoring polynomials over finite fields with
  {D}rinfeld modules}.
\newblock \bibinfo{journal}{\emph{Math. Comp.}}  \bibinfo{volume}{73}
  (\bibinfo{year}{2004}), \bibinfo{pages}{317--322}.
\newblock


\bibitem[\protect\citeauthoryear{Wiedemann}{Wiedemann}{1986}]%
        {Wiedemann:1986:SSL:13738.13744}
\bibfield{author}{\bibinfo{person}{D~H Wiedemann}.}
  \bibinfo{year}{1986}\natexlab{}.
\newblock \showarticletitle{Solving Sparse Linear Equations over Finite
  Fields}.
\newblock \bibinfo{journal}{\emph{IEEE Trans. Inf. Theor.}}
  \bibinfo{volume}{32}, \bibinfo{number}{1} (\bibinfo{year}{1986}),
  \bibinfo{pages}{54--62}.
\newblock
\showISSN{0018-9448}


\end{thebibliography}

\end{document}